\RequirePackage{fix-cm}
\documentclass[smallextended,numbook]{svjour3-pp}       
\smartqed  
\usepackage{natbib, graphicx}
\usepackage{amsmath, amsfonts, amssymb}
\usepackage{enumerate}

\journalname{Journal of Mathematical Biology}

\allowdisplaybreaks

\begin{document}

\title{Selective sweeps for recessive alleles and for other
  modes of dominance }

\titlerunning{Selective sweeps for recessive alleles}        

\author{G. Ewing \and J. Hermisson \and\\ P. Pfaffelhuber \and
  J. Rudolf}


\institute{G. Ewing \at
University of Vienna\\
\email{greg.ewing@univie.ac.at}           
\and J. Hermisson \at 
University of Vienna\\ 
\email{joachim.hermisson@univie.ac.at}
\and P. Pfaffelhuber\at 
Albert-Ludwigs University of Freiburg\\
\email{p.p@stochastik.uni-freiburg.de}
}

\date{Received: date / Accepted: date}

\maketitle

\begin{abstract}
  \noindent
  A selective sweep describes the reduction of linked genetic
  variation due to strong positive selection. If $s$ is the fitness
  advantage of a homozygote for the beneficial allele and $h$ its
  dominance coefficient, it is usually assumed that $h=1/2$, i.e.\ the
  beneficial allele is co-dominant. We complement existing theory for
  selective sweeps by assuming that $h$ is any value in $[0,1]$. We
  show that genetic diversity patterns under selective sweeps with
  strength $s$ and dominance $0<h<1$ are similar to co-dominant sweeps
  with selection strength $2hs$. Moreover, we focus on the case $h=0$
  of a completely recessive beneficial allele. We find that the length
  of the sweep, i.e.\ the time from occurrence until fixation of the
  beneficial allele, is of the order of $\sqrt{N/s}$ generations, if
  $N$ is the population size. Simulations as well as our results show
  that genetic diversity patterns in the recessive case $h=0$ greatly
  differ from all other cases.  
\end{abstract}

\section{Introduction}
The model of selective sweeps (also called genetic hitchhiking)
predicts a reduction in sequence diversity at a neutral locus closely
linked to a beneficial allele \citep{MaynardSmithHaigh1974}. Most
analysis of this model assumes that the beneficial allele is co-dominant. 
Accordingly, genome scans for the evidence of recent positive selection 
test a neutral model against (strong) selection for a co-dominant allele.
Many such methods use information about the expected site frequency spectrum 
under a hitchhiking model to detect signatures of positive selection (e.g.\
\citealp{KimStephan2002,NielsenEtAl2005,JensenEtAl2005}). 
Simpler approaches use test statistics such as sample heterozygosity 
(usually called $\pi$), or Tajima's $D$ to reject a standard neutral model. 
Simulation results by \cite{pmid16219788} and \cite{pmid16687733} show that the
false-discovery and false-negative rates of such methods increase if selection acts 
on a recessive rather than a co-dominant allele.

Although adaptations are often assumed to be rather dominant than recessive 
\citep{Charlesworth1998}, 
also the case of recessive beneficial alleles is well documented in the
empirical literature. Many cases that have been described concern resistance
alleles. Here, the loss of a function of a gene conveys resistance to a
pathogen. Often only the homozygous mutant is resistant, leading to a recessive 
sweep. Examples include: (i) The Duffy blood group locus in humans, where
the homozygous null-allele (FY-0) confers complete resistance to vivax malaria. 
\cite{HamblinDiRienzo2000} report that the FY-0-genotype is
at or near fixation in most sub-Saharan African populations but is
very rare outside Africa, which suggests that it is locally under strong
positive selection.
(ii) Resistance to the yellow mosaic virus disease in barley has been
mapped to several recessive resistance genes
\citep{pmid15131346}. (iii) The plant gene eIF4E (present e.g.\ in
pepper, pea and tomato) is a factor involved in basic cellular
processes and can be used by viruses to complete their life
cycle. Only if the function of both gene copies is compromised, the
plant is resistant and positive selection can act
\citep{pmid18953590}. (iv) The yellow fever mosquito \emph{{A}edes
  aegypti} is resistant to the drug permethrin, if both copies carry a
replacement mutation in the gene para, as shown in \cite{Garcia2009}.

Starting with the original publication by
\cite{MaynardSmithHaigh1974}, an extensive body of analytical theory
has been established for the hitchhiking model (e.g.\
\citealp{KaplanHudsonLangley1989, StephanWieheLenz1992, Barton1998,
  DurrettSchweinsberg2004, EtheridgePfaffelhuberWakolbinger2006}).  In
addition to results on reduced diversity and the frequency spectrum,
linkage disequilibrium has been studied by
\cite{StephanSongLangley2006,McVean2007,PfaffelhuberStudeny2007,
  LeocardPardoux2010}. Moreover, the model was extended to the case of
multiple origins of the beneficial allele due to mutation or from
standing genetic variation \citep{PenningsHermisson2006a,
  HermissonPennings2005, PenningsHermisson2006b,
  HermissonPfaffelhuber2008}. All these results are built around the
simplest possible scenario for adaptation, where positive selection
acts on a single locus without dominance. Despite its empirical
importance, dominance was only studied quite
recently. \cite{pmid16219788} use computer simulations to demonstrate
the impact of intermediate dominance on the most important summary
statistics for the frequency spectrum, see also
\cite{pmid16687733}. Explicit analytical results are even more sparse
and only exist for the fixation time (duration of the sweep) of the
beneficial allele \citep{vanHerwaarden2002}. The case of a completely
recessive beneficial allele, in particular, has not been treated in
any of these publications.

The goal of our investigation is the extension of previous analytical
results to the case of arbitrary dominance. We focus, in particular,
on the completely recessive case.  While sweeps with intermediate
dominance, $0<h<1$, and selection coefficient $s$ produce diversity
patterns similar to a co-dominant beneficial allele with selection
coefficient $2hs$ (Theorem \ref{T2h}), our results show that recessive
sweeps, $h=0$, are qualitatively different (Theorem \ref{T2}). Also,
for the probability of fixation and the duration of the sweep, the
recessive case is widely different from other modes of dominance. See
Proposition \ref{P1} and Theorem \ref{T1} for the recessive case and
Proposition \ref{P2} and Theorem \ref{T3} for $0<h\leq 1$.

The paper is organized as follows: in Section \ref{S:2} we introduce
the model for selective sweeps with arbitrary dominance coefficient,
both at the selected locus (Subsection \ref{S:a}) and at the neutral
locus (Subsection \ref{S:defSC}). In Section \ref{S:results}, we give
our results on sweeps of recessive alleles.
Section \ref{S:results2} contains our results for sweeps in the cases
$0<h\leq 1$. In Section \ref{S:div}, we describe sequence diversity
patterns under recessive sweeps using simulations and compare them
with the case $0<h< 1$. We conclude with the proofs in Section
\ref{S:proofs}.

\section{The model}
\label{S:2}
We use discrete (Wright--Fisher) models as well as diffusion processes
for modeling allelic frequency paths (Section \ref{S:a}). In order to
study genetic diversity patterns, we use a structured coalescent
(Section \ref{S:defSC}).

\subsection{The allelic frequency path}
\label{S:a}
Consider a one-locus Wright--Fisher model, consisting of $N$ diploid
(and hence $2N$ haploid) individuals. The beneficial mutant is $B$ and
the wildtype allele is denoted $b$. The (relative) fitness of diploids
is given as follows:

\begin{center}
\begin{tabular}{lccc}
  Genotype & $BB$ & $Bb$ & $bb$ \\
  Relative fitness & $1+s$ & $1+sh$ & $1$
\end{tabular}
\end{center}

We are interested in the dynamics of $(X_t^N)_{t=0,1,2,...}$, where
$X^N_t$ is the frequency of the beneficial allele $B$ in generation
$t$. This process is a Markov chain and given $X^N_t=\tfrac i{2N}$,
the transition probabilities are 
$$ \mathbb P[X^N_{t+1} = \tfrac j{2N}|X^N_t=\tfrac i{2N}] = \binom{2N}{j} 
\tilde p_i^j(1-\tilde p_i)^{2N-j},$$ where
$$ \tilde p_i = \frac{i^2(1+s) + i(2N-i)(1+sh)}{i^2(1+s) + 
  2i(2N-i)(1+sh)+(2N-i)^2}.  $$

~

\noindent
For $N\to\infty$, $s\to 0$ such that $2Ns\to\alpha$, and $X^N_0
\Rightarrow X_0$, it is well-known (e.g. \citealp{Ewens2004}) that
$(X^N_{\lfloor Nt\rfloor})_{t\geq 0} \Rightarrow (X_t)_{t\geq 0}$,
where '$\Rightarrow$' means convergence in distribution (in the space
of real-valued functions, equipped with uniform convergence on
compacts), where $\mathcal X := (X_t)_{t\geq 0}$ is the diffusion,
uniquely determined by the stochastic differential equation (SDE)
given in \eqref{eq:SDE0}.

\begin{definition}[Allelic frequency path]\label{def:X}
  The diffusion $\mathcal X = (X_t)_{t\geq 0}$ is the unique solution
  of the SDE
  \begin{align}\label{eq:SDE0}
    dX = \alpha(h+X(1-2h))X(1-X)dt + \sqrt{X(1-X)}dW,
  \end{align}
  where $W$ is a standard Brownian motion. We set
  \begin{align} \label{eq:stop} T_0:=\inf\{t\geq 0: X_t=0\}, \qquad
    T_1:=\inf\{t\geq 0: X_t=1\},
  \end{align}
  which are the times of loss and fixation of the beneficial allele,
  respectively. Moreover, $\mathcal X^{\ast} = (X_t^{\ast})_{t\geq 0}$
  is the process $\mathcal X$, conditioned on the event $\{T_1<T_0\}$.
  We set
  \begin{align}
    \label{eq:Tast}
    T^\ast := \inf\{t\geq 0: X_t^{\ast} = 1\},
  \end{align}
  which is the time of fixation of the beneficial allele.  If not
  mentioned otherwise, we assume that $X_0 = 0$ and $\mathcal X^\ast$
  arises as limit of conditioned processes which are started in
  $\varepsilon$ as $\varepsilon\to 0$. For $\alpha, h\in\mathbb R$, we
  denote the distribution of $\mathcal X$, started in $X_0=x$, by
  $\mathbb P_x^{\alpha,h}[.]$. Expectations and variances are denoted
  by $\mathbb E_x^{\alpha,h}[.]$ and by $\mathbb
    V_x^{\alpha,h}[.]$, respectively.

%
\end{definition}

\begin{remark}[Fixation probability for a single mutant]
  \label{rem:conv1}
  \sloppy Let $\mathbb P_{x}^{N,s,h}[.]$ be the probability measure
  for the Wright--Fisher model with population size $N$, selection
  coefficient $s$ and dominance $h$, started in $X_0^N=x$. Note that
  the weak convergence $(X^N_{\lfloor Nt\rfloor})_{t\geq 0}
  \Rightarrow (X_t)_{t\geq 0}$ does not imply convergence for all
  interesting functionals of the Wright--Fisher model. In particular,
  convergence of fixation probabilities for a single mutant in the
  sense
  $$ \frac{\mathbb P_{1/(2N)}^{N,s_N,h}[X_\infty^N=1]}{
    \mathbb P_{1/(2N)}^{\alpha,
      h}[X_\infty=1]}\xrightarrow{N\to\infty} 1$$ with
  $Ns_N\xrightarrow{N\to\infty}\alpha$ has only been proved in the
  case $h=\tfrac 12$ (see \citealp[p. 565]{BuergerEwens1995}). For
  this reason, all our assertions are stated in the diffusion
  framework.
\end{remark}

\subsection{The structured coalescent}
\label{S:defSC}
Consider a sample of size $n$, taken from the population at the time
of fixation of the beneficial allele.  From each of these individuals,
consider homologous neutral loci, which are linked to the beneficial
allele at recombination distance $\rho$ (i.e.\ the probability that a
recombination occurs between the beneficial and the neutral locus is
$r$ per generation in the Wright--Fisher model of diploid size $N$, and
$2N r \to \rho$ as $N\to\infty$). In order to study the genetic
diversity within these neutral loci, we follow
\cite{KaplanHudsonLangley1989} and \cite{BartonEtheridgeSturm2004},
and introduce the following coalescent process, which is conditioned
on a path of $\mathcal X^\ast$.

\begin{definition}[The structured coalescent \boldmath$\mathcal K$]
  \label{def:K}
  Let $\mathcal X^\ast$ be as in Definition \ref{def:X}. Conditioned
  on $\mathcal X^\ast$, we define a Markov process, $\mathcal K =
  (\mathcal K^b, \mathcal K^w) = (K_\beta^b, K_\beta^w)_{0\leq
    \beta\leq T^\ast}$ with $\beta = T^\ast-t$. Here, $K_\beta^b$ and
  $K_\beta^w$ are the number of lines in the beneficial and wild-type
  background at time $\beta$, respectively. Taking values in $\mathbb
  Z_+^2$, this process starts in $(K_0^b, K_0^w) = (n,0)$ for some
  $n\in\mathbb N$. If $K_\beta = (k^b, k^w)$, then there are
  transitions to
  \begin{enumerate}
  \item $(k^b-1, k^w)$ at rate $\binom{k^b}{2}\tfrac{1}{X_{T^\ast -
        \beta}^\ast}$\\ (two lines in the beneficial background are
    merged to a single line)
  \item $(k^b, k^w-1)$ at rate $\binom{k^w}{2}\tfrac{1}{1-X_{T^\ast
        -\beta}^\ast}$ \\ (two lines in the wildtype background are
    merged to a single line)
  \item $(k^b-1, k^w+1)$ at rate $k^b  \rho (1-X_{T^\ast
      -\beta})$ \\ (one line in the beneficial background changes to
    the wildtype background)
  \item $(k^b+1, k^w-1)$ at rate $k^w \rho X_{T^\ast - \beta}$ \\
    (one line in the wildtype background changes to the beneficial
    background).
  \end{enumerate}
  See Figure \ref{fig1} for an illustration of $\mathcal K$.
\end{definition}

\begin{figure}
  \vspace*{-1.5cm}

  \centering
  \includegraphics[width=10cm]{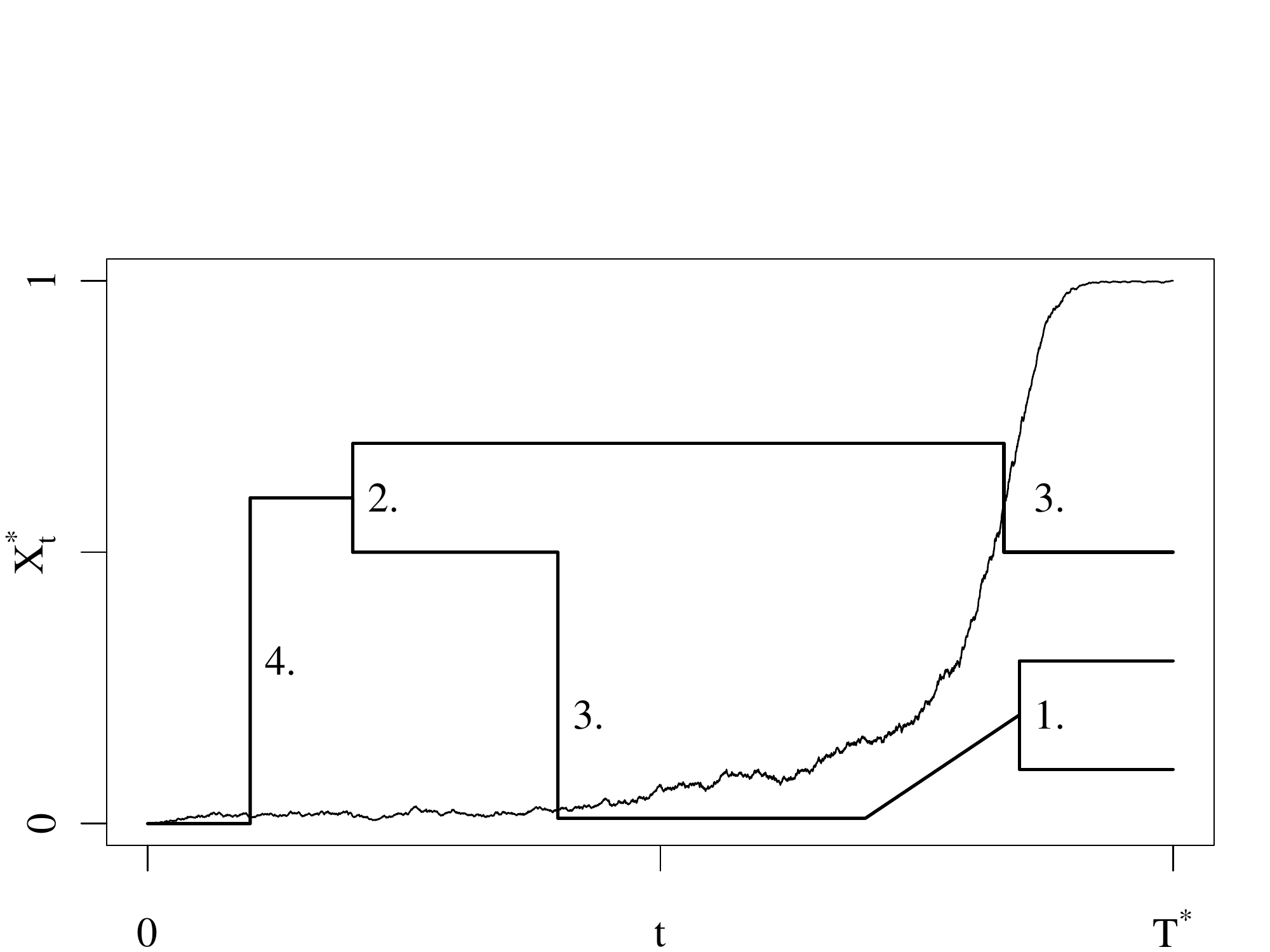} 
  \caption{\label{fig1} Illustration of the structured coalescent in
    the case $h=0$. Pairs of lines can coalesce within the beneficial
    background (event 1.), or within the wildtype background (event
    2.). Changes from the beneficial to the wildtype background (event
    3.) or vice versa (event 4.) may occur as well.}
\end{figure}

\section{Results on recessive alleles}
\label{S:results}
In this section, we focus on the case $h=0$, i.e.\ on properties of
the diffusion $\mathcal X = (X_t)_{t\geq 0}$ given by the SDE
\begin{align}\label{eq:SDE}
  dX = \alpha X^2(1-X)dt + \sqrt{X(1-X)}dW
\end{align}
and the corresponding diffusion $\mathcal X^\ast$, which is
conditioned to hit 1 (recall from Definition \ref{def:X}). It is
crucial to note that the process $\mathcal X^\ast$ spends most of its
time at low frequencies for $h=0$; see Figure \ref{fig:path}(A). The
reason is that for low frequency, most beneficial alleles are found in
heterozygotes and selection can hence not be efficient. In order to
make this statement more quantitative, we will show (see Lemma
\ref{lkey}), that $(\sqrt\alpha X_{t/\sqrt\alpha})_{t\geq 0}$
converges to the diffusion $(Y_t)_{t\geq 0}$ given by $dY = Y^2 dt +
\sqrt{Y} dW$. This implies that $\mathcal X^\ast$ spends most of its
time in frequencies of order $1/\sqrt\alpha$. See also Figure
\ref{fig:path}(B) for an illustration of the process $(\sqrt\alpha
X_{t/\sqrt\alpha})_{t\geq 0}$.

\begin{figure}
\centering \hspace*{0cm}(A) \hspace{.45\textwidth} (B)\\ 

{\includegraphics[width=.48\textwidth,trim=0 0 50 50,clip]{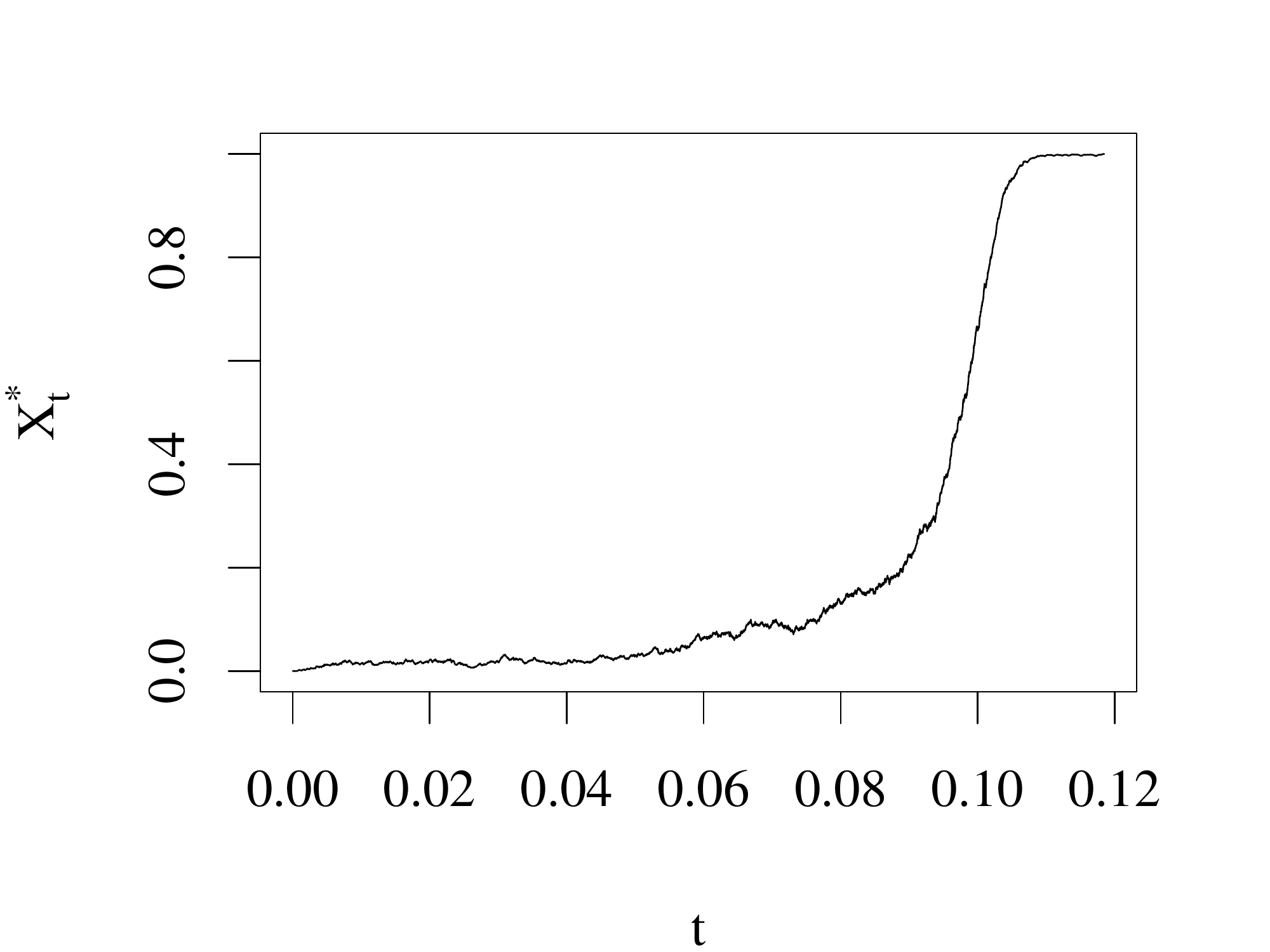}}\hfill%
{\includegraphics[width=.48\textwidth,trim=0 0 50 50,clip]{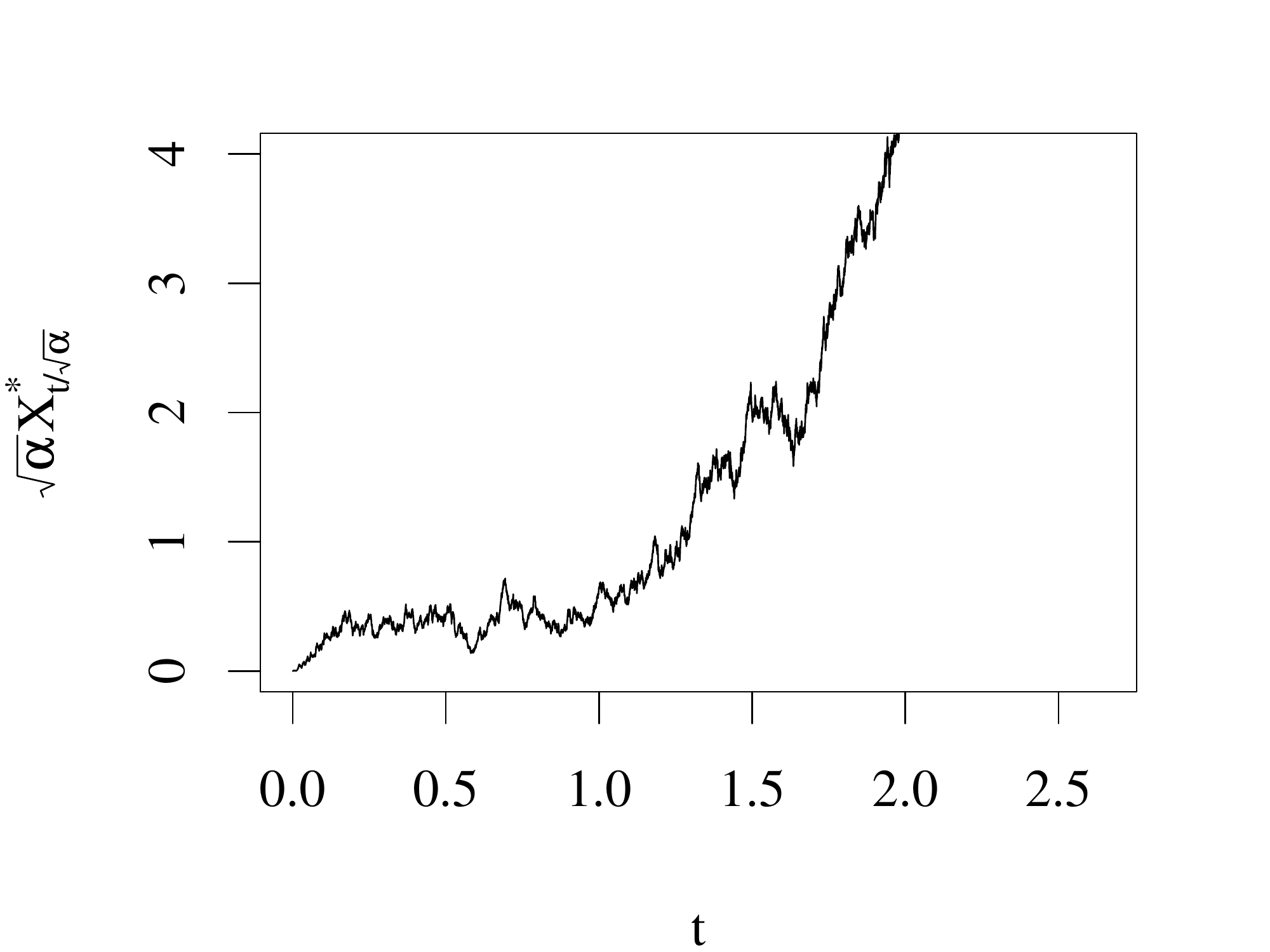}}%
\caption{\label{fig:path}%
The frequency path of a beneficial
    recessive allele, conditioned on fixation, $\mathcal X^\ast$. We
    have used a Wright--Fisher model with $N=5\cdot 10^4$ and $s=5\cdot
    10^{-3}$, i.e.\ $2Ns = 500$. The right figure shows the rescaled
    frequency path $(\sqrt\alpha X^\ast_{t/\sqrt\alpha})_{t\geq 0}$.}
\end{figure}

We give our main three results on fixation probability (Proposition
\ref{P1}) -- actually already derived by \cite{Kimura1962} -- the
duration of the recessive sweep (Theorem \ref{T1}) and the structured
coalescent (Theorem \ref{T2}). All proofs are found in Section
\ref{S:proofs1}.

\subsection{Fixation probability}
The following is a classical fact, which can be read off from
\cite{Kimura1962}, equation (14). It describes the fixation
probability of the beneficial allele and is stated here for
completeness.

\begin{proposition}[Fixation probability]\label{P1}
  Let $\mathcal X$ be as in Definition \ref{def:X}, $\alpha\to\infty$
  and $\varepsilon_\alpha$ be such that $\varepsilon_\alpha 
  \sqrt\alpha \to 0$. Then,\footnote{ For two functions $a_\alpha$ and
    $b_\alpha$ we write $a_\alpha \stackrel{\alpha\to\infty}\approx
    b_\alpha$ iff $ \lim_{\alpha\to\infty} \frac{a_\alpha}{b_\alpha} =
    1$.}
  \begin{align}
    \mathbb P_{\varepsilon_\alpha}^{\alpha, h=0}[T_1<T_0]
    \stackrel{\alpha\to\infty}\approx 2 \varepsilon_\alpha 
    \sqrt{\frac{\alpha}{\pi}}.
  \end{align}
\end{proposition}

\begin{remark}[Fixation probability in a finite population]
  Consider a finite population of diploid size $N$. Assume that $X_0 =
  \tfrac{1}{2N}$, meaning that there is only a single copy of the
  beneficial allele at time $0$. For $\alpha=2Ns$, we find that, as
  $N\to\infty$,
  $$\mathbb P_{1/(2N)}^{\alpha, h=0}\big[T_1<T_0\big] 
  \approx \frac{1}{N}\sqrt{\frac{2Ns}{\pi}} = \sqrt{\frac{2s}{\pi
      N}}.$$ (This is exactly equation (15) in
  \citealp{Kimura1962}). Hence, a new recessive beneficial allele has
  a chance to be fixed in the population, which is much larger than
  $1/{2N}$, its chance if it was neutral, and much smaller than $2hs$,
  its chance if it would not be recessive (compare with Remark
  \ref{rem:42}).  However, note that it is not shown that the fixation
  probability for a single copy of the beneficial allele in a
  Wright--Fisher model has the same limit behavior; compare with Remark
  \ref{rem:conv1}.
\end{remark}

\subsection{Length of the recessive sweep}
Now, we come to the results on the duration of the recessive sweep.

\begin{pptheorem}[Length of the recessive sweep]
  \label{T1}
  Let $T^\ast$ be as in \eqref{eq:Tast}. Then,
  \begin{align}\label{eq:T11a}
    \mathbb E_0^{\alpha, h=0}[T^\ast] - \frac{4
      c_{cat}}{\sqrt{\pi}} \cdot \frac{1}{\sqrt{\alpha}} &
    \stackrel{\alpha\to\infty}\approx \frac{3\log\alpha}{2\alpha},
    \\\label{eq:T12b} \mathbb{V}_0^{\alpha, h=0}[T^\ast] &
    \stackrel{\alpha\to\infty}\approx \frac{c}{\alpha},
  \end{align}
  where $c_{cat} \approx 0.916$ is Catalan's constant and some
  $0<c<\infty$.
\end{pptheorem}

\begin{remark}[Further investigation of \boldmath $T^\ast$]
  \begin{enumerate}
  \item   We find from \eqref{eq:T11a} that
    $$\mathbb E_0^{\alpha, h=0}[T^\ast] \approx \frac{2.067}{\sqrt\alpha} + \frac{3\log\alpha}{2\alpha}$$
    for large values of $\alpha$. For a finite population of size $N$
    and a beneficial recessive allele with selection coefficient $s$,
    this means that it takes $2.067\sqrt{N/s}+3\log(2Ns)/2s$
    generations until fixation of the beneficial allele on average.
    As simulations show, this gives accurate numerical results; see
    Figure \ref{fig:length}(A). In addition, a numerical investigation
    of \eqref{eq:var} reveals that
    \begin{align}
      \mathbb V_0^{\alpha, h=0}[T^\ast] \approx \frac{0.6362}{\alpha}
    \end{align}
    for large $\alpha$. This value also fits well with simulations;
    see Figure \ref{fig:length}(B).
  \item Our calculations reveal not only limits for the expected total
    duration of the sweep, but also of the duration of the allele
    being in low (below some $\varepsilon>0$, not depending on
    $\alpha$), intermediate (between $\varepsilon$ and
    $1-\varepsilon$) and in high (above $1-\varepsilon$)
    frequency. These results are collected in Table \ref{tab:1}.
  \end{enumerate}
\end{remark}

\begin{table}
  \begin{center}
    \begin{tabular}{ l||c|c|c|} 
      \rule[-4mm]{0cm}{1cm}& $0\leq X_t\leq \varepsilon$ &  $\varepsilon<X_t \leq 1-\varepsilon$ & 
      $1-\varepsilon<X_t \leq 1$ \\ \hline\hline
      \rule[-4mm]{0cm}{1cm}$h=0$ & $\displaystyle\frac{4 c_{\text{cat}}}{\sqrt{\pi\alpha}} + 
      \frac{\log\alpha}{2\alpha}$ & 
      $\displaystyle\mathcal O\Big(\frac 1 \alpha\Big)$ 
      & $\displaystyle\frac{\log\alpha}{\alpha}$ \\\hline
      \rule[-4mm]{0cm}{1cm}$0<h<1$ & $\displaystyle\frac{\log\alpha}{h\alpha}$ & $\displaystyle
      \mathcal O\Big(\frac 1 \alpha\Big)$ 
      & $\displaystyle\frac{\log\alpha}{(1-h)\alpha}$ \\\hline
      \rule[-4mm]{0cm}{1cm}$h=1$ & $\displaystyle\frac{\log\alpha}{\alpha}$ & $\displaystyle\mathcal O\Big(
      \frac 1 \alpha\Big)$ 
      & $\displaystyle\frac{\pi^{3/2}}{2\sqrt\alpha} + \frac{\log\alpha}{2\alpha}$ \\\hline
    \end{tabular}
  \end{center}
  \caption{\label{tab:1}The leading terms for the expected duration of
    the time when the frequency of the beneficial allele is low
    (i.e. smaller than some $\varepsilon>0$), intermediate (between
    $\varepsilon$ and $1-\varepsilon$) and high (above
    $1-\varepsilon)$. In case the leading term can only be determined up
    to a factor, we write $\mathcal O(.)$.}
\end{table}

\begin{figure}
  \centering \hspace*{0cm}(A) \hspace{.45\textwidth} (B)\\ 

{\includegraphics[width=.48\textwidth,trim=0 0 30 30,clip]{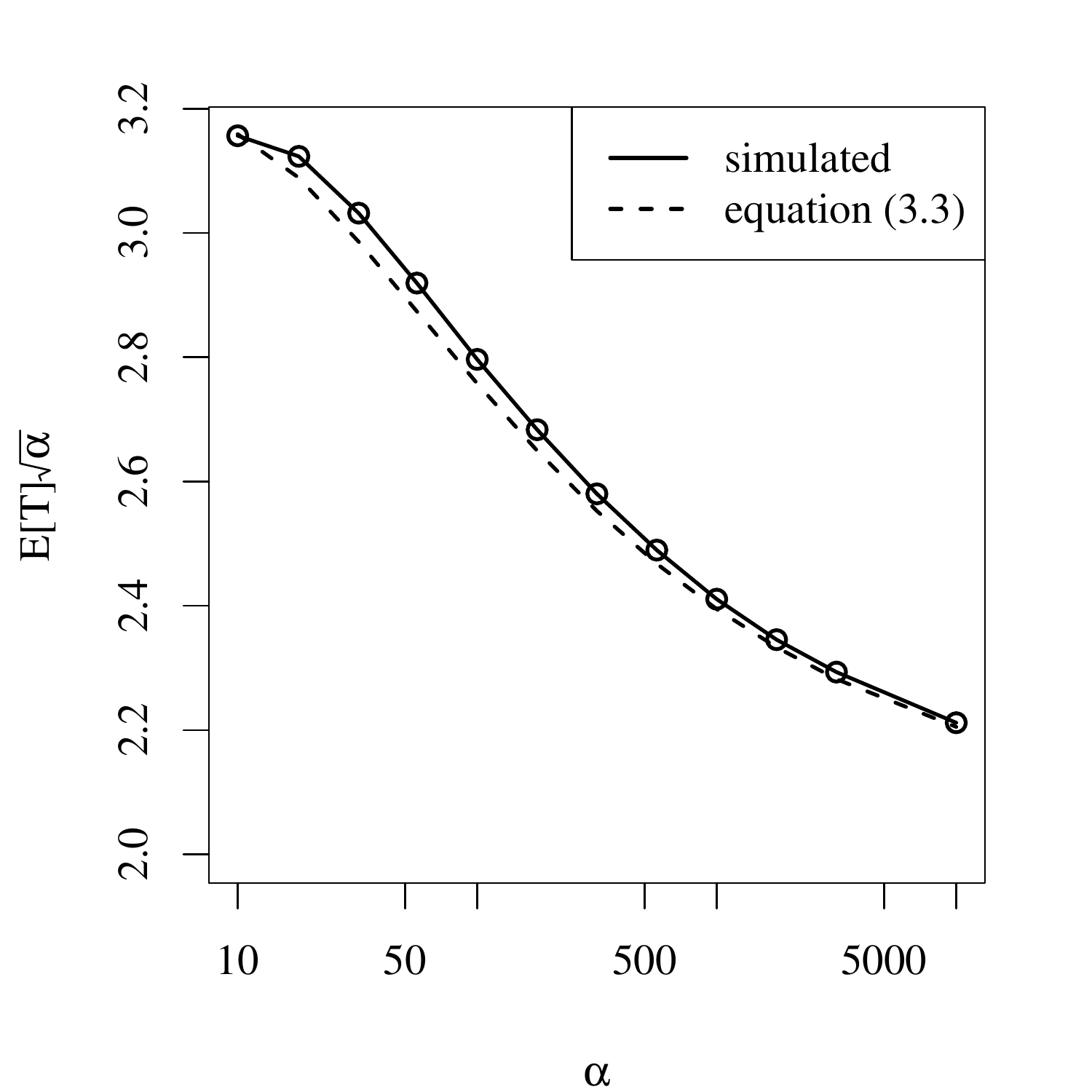}}\hfill
{\includegraphics[width=.48\textwidth,trim=0 0 30 30,clip]{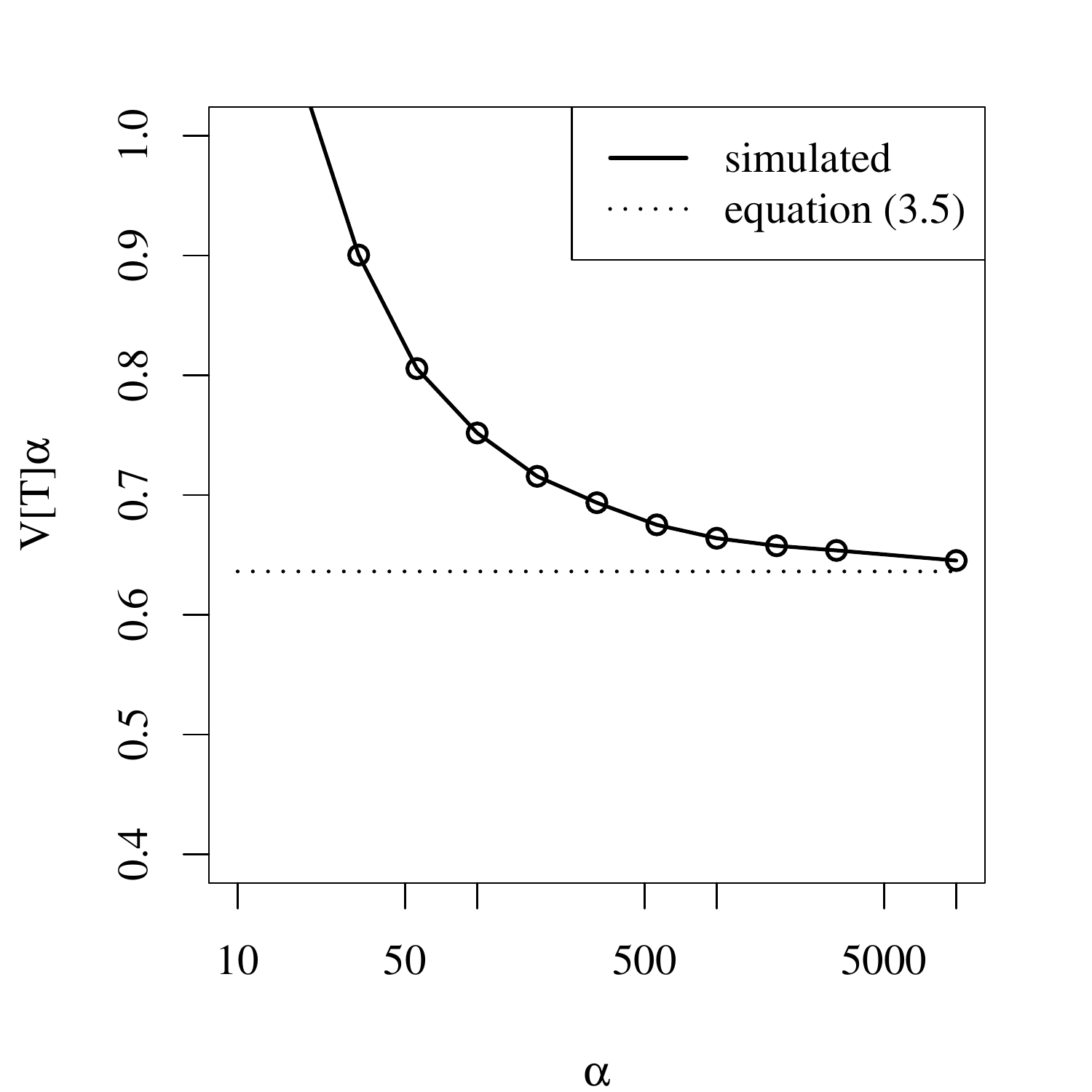}}
  \caption{\label{fig:length}The duration of the sweep of a beneficial
    recessive allele, is plotted against $\alpha=2Ns$. We have used a
    Wright--Fisher model, with $N=10^6$. (A) Average length of the
    duration of the recessive sweep. (B) Empirical variance of the
    duration of the recessive sweep.}
\end{figure}

\begin{remark}[Deterministic approach]
  In the co-dominant case, the fixation process is often approximated
  by a logistic increase of the beneficial allele;
  e.g. \cite{KaplanHudsonLangley1989, StephanWieheLenz1992}. Formally,
  this approximation replaces the SDE (\ref{eq:SDE0}) by the
  corresponding ODE. Calculating the duration of the sweep as the time
  to reach $1-\tfrac{1}{2N}$ from $\tfrac{1}{2N}$ under the logistic
  results is a rough, but acceptable first approximation for the
  fixation time. This is different for the completely recessive case,
  where the corresponding ODE reads $$ \dot y = \alpha y^2(1-y),
  \qquad y_0 = \tfrac{1}{2N}.$$ The time $t^\ast$ to reach $y_{t^\ast}
  = 1-\tfrac{1}{2N}$ from an initial frequency $\tfrac{1}{2N}$ is
  easily derived to be
  $$ t^\ast = \frac{1}{\alpha} \Big( \frac{4N(N-1)}{2N-1} + 2\log(2N-1) \Big)
  \stackrel{N\to\infty}\approx \frac{2N}{\alpha}.$$ Note that the asymptotic behavior for $\alpha \to \infty$
  of this expression differs considerably from the scaling $T^* \sim \alpha^{-1/2}$ in \eqref{eq:T11a} that we have
  found from the full stochastic equation.
\end{remark}

\subsection{The structured coalescent}
In this section, we study the structured coalescent from Section
\ref{S:defSC} for $h=0$.

\begin{pptheorem}[The structured coalescent for \boldmath $h=0$]
  \label{T2} ~\newline
  Let $\mathcal K = (\mathcal K^b_\beta, \mathcal K^w_\beta)_{0\leq
    \beta \leq T^\ast}$ be as in Definition \ref{def:K}, and $\rho =
  \rho_\alpha$ such that $\rho/\sqrt\alpha \to \lambda$ for some
  $0<\lambda<\infty$ as $\alpha\to\infty$. Then,
  \begin{align} \label{T21} \mathbb P_0^{\alpha, h=0}[\mathcal
    K_{T^\ast} = (1,0) | \mathcal K_0 = (1,0)] &
    \stackrel{\alpha\to\infty}\approx \mathbb E_0^{\alpha,
      h=0}[e^{-\lambda (\sqrt\alpha\cdot T^\ast)}],\\ \label{T22}
    \mathbb P_0^{\alpha, h=0}[K^b_{T^\ast} + K^w_{T^\ast} = k |
    \mathcal K_0 = (n,0)] & \stackrel{\alpha\to\infty}\approx
    c_{\lambda,n,k}
  \end{align}
  for some $0<c_{\lambda,n,k}<1$, not depending on $\alpha$.
\end{pptheorem}

~

\begin{remark}[Approximation of the genealogy for \boldmath $h=0$]
  \label{rem:34}
  It is important to note that only the scaling of the recombination
  rate $\rho$ by $\sqrt\alpha$ leads to a non-trivial limit result in
  \eqref{T21} and \eqref{T22}. In applications, however, finite values
  of $\alpha$ and $\rho$ must be assumed. In order to use the above
  Theorem for large $\alpha$, we set $\lambda=\rho/\sqrt\alpha$. In
  this case, the Theorem implies that every single line changes from
  the beneficial to the wildtype background approximately at rate
  $\rho = \lambda\sqrt\alpha$ -- see \eqref{T21} -- during the
  sweep. A first naive guess about the genealogy at the neutral locus
  is a star-like approximation as in the case of co-dominant alleles
  (\citealp{DurrettSchweinsberg2004}, see also Section
  \ref{S:str2}). This means that all lines recombine to the wild-type
  background (transition 3.)  independently at rate $\rho$ and those
  which did not recombine coalesce at the beginning of the
  sweep. However, $\mathcal X^\ast$ spends most time in frequencies of
  order $1/\sqrt\alpha$, leading to an increased rate of coalescence
  during the sweep. We suggest that the process $\mathcal K$ can be
  approximated by a coalescence process (which is not conditioned on
  $\mathcal X^\ast$) where two lines merge at rate $\sqrt\alpha$ and
  each line escapes the sweep at rate $\rho$. As we will see in
  Section \ref{S:div}, this approximation (in contrast to the
  star-like approximation) produces reasonable numerical results and
  helps to explain some of the qualitative differences in frequency
  spectra of recessive and co-dominant sweeps. Note, however, that the
  approximation is not a formal convergence result. (Due to the
  fluctuations in $\mathcal X^\ast$ around $1/\sqrt\alpha$ such a
  formal result is not easily obtained).
\end{remark}

\begin{remark}[Reduction of diversity and \boldmath $c_{\lambda,2,2}$]
  \label{rem:35}
  The most prominent attribute of a selective sweep is
  the reduction in sequence diversity, which can be quantified by
  the average reduction in heterozygosity due to the sweep. 
  Let $H_2(t)$ be the heterozygosity at a neutral locus linked to
  the selected site. Then (see equation (16) in \citealp{KaplanHudsonLangley1989})
  $$  \mathbb P_0^{\alpha, h=0}[H_2(T^\ast)] = p_{22} \cdot \mathbb
  P[H_2(0)],$$ for
  $$  p_{22}:=\mathbb P_0^{\alpha, h=0}[K^b_{T^\ast} + K^w_{T^\ast} = 2 |
  \mathcal K_0 = (2,0)] \approx c_{\lambda, 2,2}.
  $$ 
  Indeed, assuming that there are no new mutations in $[0,T^\ast]$,
  two neutral loci picked at the end of the selective sweep are
  different if the ancestors at the beginning of the sweep are
  different (probability $p_{22}$) and if these carry different
  alleles (probability $\mathbb P[H_2(0)]$).  In particular, $p_{22}$
  captures the reduction of diversity within recessive sweeps. Using
  the approximate genealogy suggested in the last remark and a
  competing Poisson process argument,
  we find
  \begin{align}\label{eq:rea}
    p_{22} \approx \frac{2\rho}{\sqrt\alpha + 2\rho} \approx
    \frac{2\lambda}{1+2\lambda}.
  \end{align}
  We will see in Section \ref{S:div} that this approximation produces
  a reasonable fit to simulations; see Figure \ref{fig7}(A). Note that
  the star-like approximation (described above) would lead to $p_{22}
  \approx \mathbb E[1-\exp(-2\rho T^*)]$, and hence would predict a
  reduction in sequence diversity which is weaker than seen in
  simulations (not shown, but compare with Figure \ref{fig7}(A)).
\end{remark}


\section{Results on incompletely and completely dominant alleles}
\label{S:results2}
In this section we give approximations for the fixation
probability (Proposition \ref{P2}) and the fixation time (Theorem
\ref{T3}) in the case $0<h\leq 1$.

\subsection*{Fixation probability}
The next result is the complement of Proposition \ref{P1} for the case
$0<h\leq 1$.

\begin{proposition}[Fixation probability]\label{P2}
  Let $0<h\leq 1$, and $\varepsilon_\alpha$ be such that
  $\varepsilon_\alpha  \alpha \to 0$ for $\alpha\to\infty$. Then,
  \begin{align}
    \mathbb P_{\varepsilon_\alpha}^{\alpha, h}[T_1<T_0]
    \stackrel{\alpha\to\infty}\approx 2h\alpha \varepsilon_\alpha.
  \end{align}
\end{proposition}

\begin{remark}[Fixation probability in a finite population]
  \label{rem:42}
  For a finite population of haploid size $2N$, and if the homozygote
  for the beneficial allele has selective advantage $\alpha = 2Ns$,
  the above result means that
  $$\mathbb P_{1/(2N)}^{\alpha, h}\big[T_1<T_0 \big] 
  \approx 2hs.$$ Hence, the fixation probability of an incompletely
  dominant allele with dominance $h$ and advantage $s$ is
  approximately the same as for a co-dominant allele with advantage
  $2hs$.
\end{remark}

\subsection{Length of the sweep}
Next, we give our results on the duration of the sweep for $0<h\leq
1$.

\begin{pptheorem}[Length of the sweep]
  \label{T3}
  \begin{enumerate}
  \item Let $0<h< 1$ and $T^\ast$ be as in \eqref{eq:Tast}. Then,
    \begin{align}\label{eq:T11}
      \lefteqn{\mathbb E_0^{\alpha, h}[T^\ast] - \frac{\log(2\alpha)}{\alpha
        h(1-h)} }\quad
\nonumber\\
& \stackrel{\alpha\to\infty}\approx \frac{1}{\alpha
        h(1-h)} \Big( \gamma + (2-3h)\log h + (3h-1)\log(1-h)\Big)
    \end{align}
    where $\gamma$ is Euler's constant, and
    $$ \mathbb V_0^{\alpha, h}[T^\ast] \stackrel{\alpha\to\infty}\approx
    \frac{1}{\alpha^2}\Big( \frac{c'}{h^2} +
    \frac{c''}{(1-h)^2}\Big) $$ for some $0<c',c''<\infty$, not
    depending on $h$.
  \item 
    For $h=1$,
    \begin{align}\label{eq:T12}
      \mathbb E_0^{\alpha, h=1}[T^\ast] -
      \frac{\pi^{3/2}}{2\sqrt\alpha} \stackrel{\alpha\to\infty}\approx
      \frac{3\log\alpha}{2\alpha}.
    \end{align}
    and 
    $$ \mathbb V_0^{\alpha, h=1}[T^\ast] \stackrel{\alpha\to\infty}\approx \frac{c'''}{\alpha}$$
    for some $0<c'''<\infty$.
  \end{enumerate}
\end{pptheorem}

\begin{remark}[Further investigation of \boldmath $T^\ast$]
  \begin{enumerate}
  \item The approximate expected duration of a sweep for $0<h<1$ from
    \eqref{eq:T11} has already been obtained in
    \cite{vanHerwaarden2002} (using other methods) and we only give
    these results for completeness. Comparison with numerical results
    (Figure \ref{fig:length2}) show that \eqref{eq:T11} is accurate as
    long as $h\alpha \gg 1$ and $(1-h)\alpha \gg 1$, i.e.\ dominance
    is intermediate. For nearly recessive alleles ($h\alpha$ of order
    1) or nearly dominant alleles ($(1-h)\alpha$ of order 1), however,
    $T^*$ is much better approximated by the formulas for $h=0$
    (\ref{eq:T11a}) or $h=1$ (\ref{eq:T12}).
  \item Although \eqref{eq:T11} is symmetric under the exchange of $h$
    and $(1-h)$, this does not reflect an exact identity of the
    fixation time. As noted by \cite{vanHerwaarden2002}, an exact
    symmetry of the fixation process under the Wright--Fisher diffusion
    is obtained under the map $(\alpha, h)\mapsto(-\alpha, 1-h)$,
    i.e.\ if fixation of a dominant beneficial allele is compared with
    a recessive deleterious one.  For beneficial alleles, the
    approximate symmetry under $h \mapsto (1-h)$ becomes much worse
    for nearly recessive or dominant alleles, as already observed by
    \cite{pmid16219788}. From our formulas for the fixation times of
    completely recessive and dominant alleles, \eqref{eq:T11a} and
    \eqref{eq:T12}, this asymmetry is most obvious. In particular, the
    expected length of a dominant sweep ($h=1$)
    \begin{align*}
      \mathbb E_0^{\alpha, h=1}[T^\ast] \approx
      \frac{2.7841}{\sqrt\alpha} + \frac{3\log\alpha}{2\alpha}.
    \end{align*}
    is much longer than any sweep of an allele with the same
    homozygous advantage but $1 > h \ge 0$.  Most of this time is
    spent near frequency 1. For the variance, numeric integration in
    the last line of \eqref{eq:var2} results in
    \begin{align}
      \label{eq:varhone}
      \mathbb V_0^{\alpha, h=1}[T^\ast] \approx \frac{1.818}{\alpha},
   \end{align}
   which is also larger than for $h=0$. See Figure \ref{fig:length3}
   for comparison with simulations.
 \item As for the recessive case, our calculations also reveal the
   expected times for three phases of the sweep, with the allele being
   in low, intermediate and high frequency (cf.\ Table \ref{tab:1}).
  \end{enumerate}
\end{remark}

\begin{figure}[htb!]
  \centering \hspace*{0cm}(A) \hspace{.45\textwidth} (B)\\ 
\fboxsep0pt
{\includegraphics[width=.48\textwidth,trim=0 0 32 32,clip]{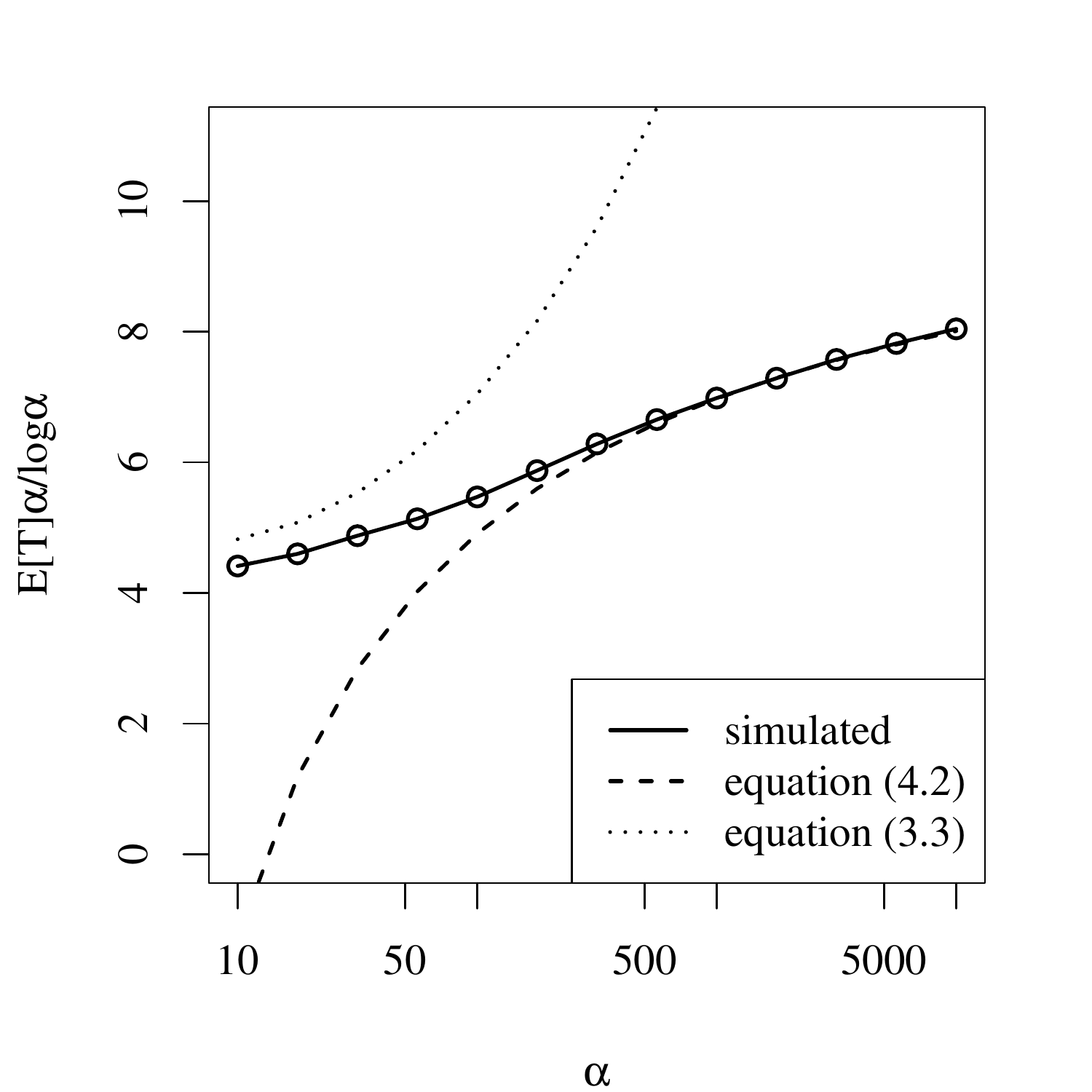}}\hfill
{\includegraphics[width=.48\textwidth,trim=0 0 32 32,clip]{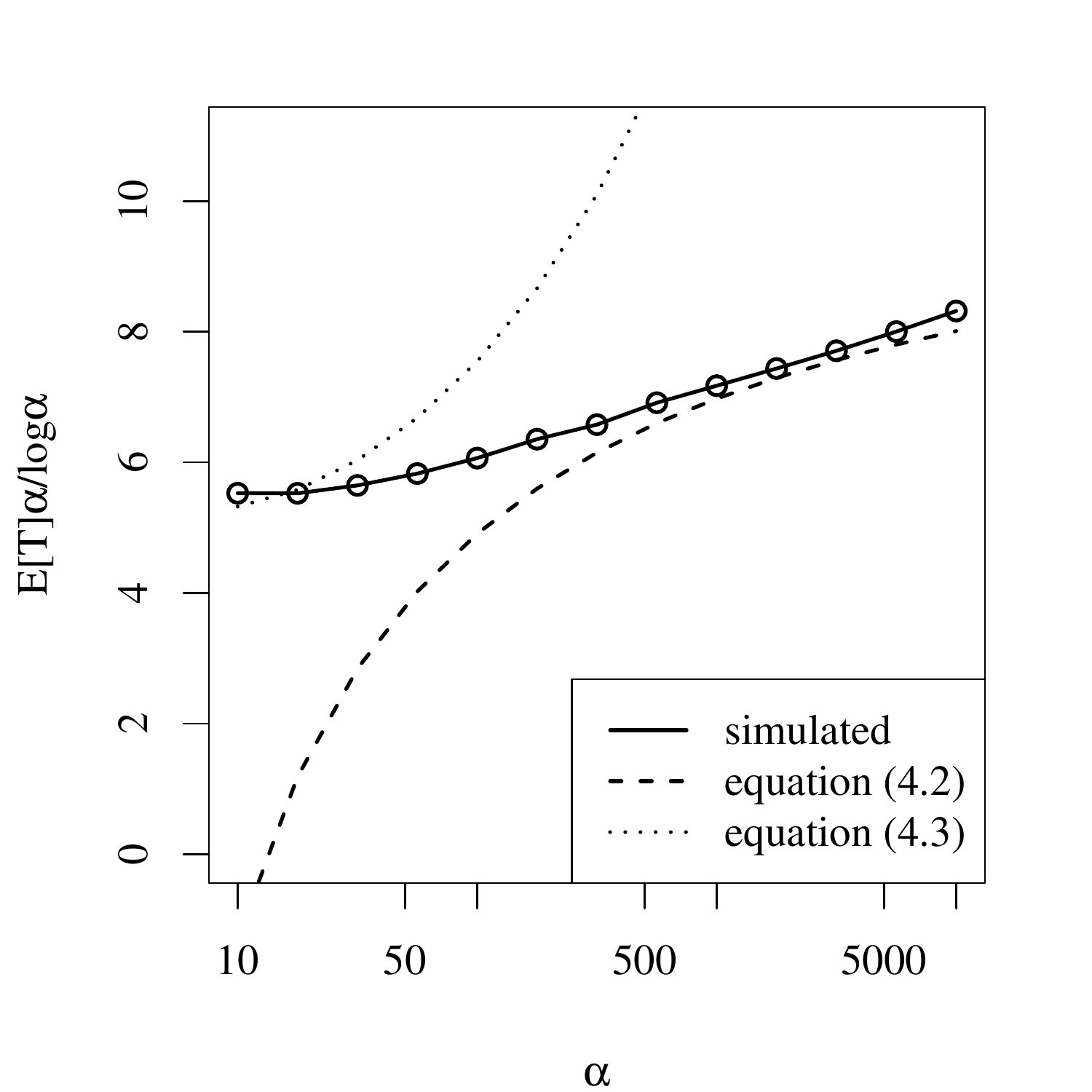}}

\caption{\label{fig:length2}%
The duration of the sweep of a
    beneficial allele is plotted against $\alpha$. We have used the
    same parameters (except for dominance) as in Figure
    \ref{fig:length}. (A) $h=0.1$, (B) $h=0.9$. For $h\alpha < 10$ (in
    (A), the nearly recessive case) and $(1-h)\alpha < 10$ (in (B),
    the nearly dominant case), the simulation curve crosses over to
    the predicted sweep times for $h=0$ \eqref{eq:T11a} and $h=1$
    \eqref{eq:T12}, respectively. }
\end{figure}

\begin{figure}[htb!]
  \begin{center}
    \centering \hspace*{0cm}(A) \hspace{.45\textwidth} (B)\\

\fboxsep0pt
{\includegraphics[width=.48\textwidth,trim=0 0 32 32,clip]{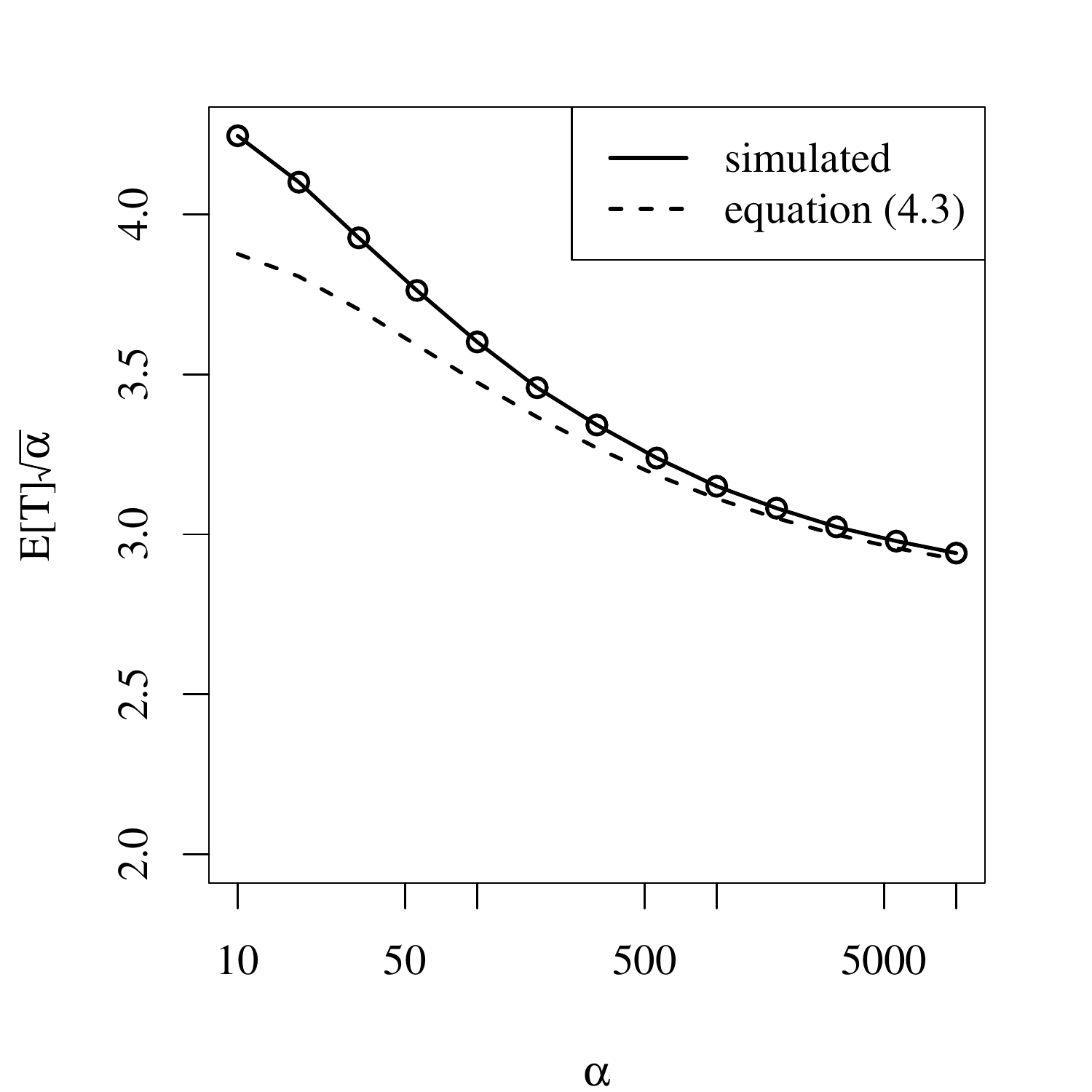}}\hfill
{\includegraphics[width=.48\textwidth,trim=0 0 32 32,clip]{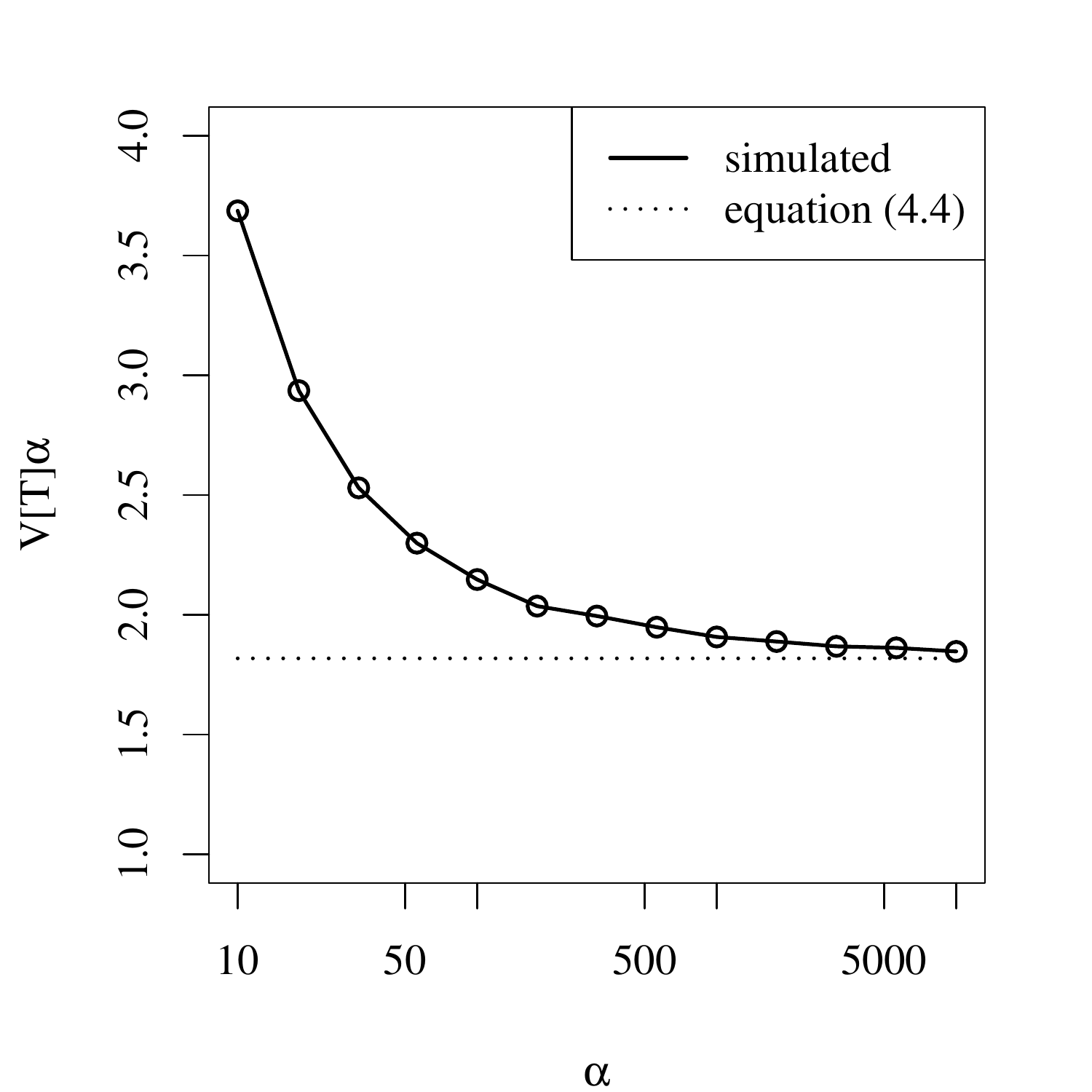}}
  \end{center}
  \caption{\label{fig:length3}(A) The duration of the sweep of a
    beneficial dominant allele, $h=1$, is plotted against
    $\sqrt\alpha$. We have used the same parameters (except for
    dominance) as in Figure \ref{fig:length}. (A) expectation, (B)
    variance.}
\end{figure}

\subsection{The structured coalescent}
\label{S:str2}
The approximation of the coalescent at the neutral locus in a sweep
region by a star-like genealogy is widely used. A formal convergence
result for the co-dominant case is due to
\cite{DurrettSchweinsberg2004}. Using our results from Section
\ref{S:proofs}, this result can be extended to the general
non-recessive case with $0<h< 1$.

~

\begin{pptheorem}[The structured coalescent for \boldmath $0<h<1$]
\label{T2h}
Let $0<h< 1$, $\mathcal K = (\mathcal K^b_\beta, \mathcal
K^w_\beta)_{0\leq \beta \leq T^\ast}$ be as in Definition \ref{def:K},
and $\rho = \rho_\alpha$ such that $\rho \log\alpha / \alpha \to
\lambda$ for some $0<\lambda<\infty$ as $\alpha\to\infty$. Then,
  \begin{align} \label{T21h} \mathbb P^{\alpha, h}[\mathcal K_{T^\ast}
    = (1,0) | \mathcal K_0 = (1,0)] &
    \stackrel{\alpha\to\infty}\approx e^{-\lambda/h},\\ \label{T22h}
    \mathbb P^{\alpha, h}[K^b_{T^\ast} + K^w_{T^\ast} = 2 | \mathcal
    K_0 = (2,0)] & \stackrel{\alpha\to\infty}\approx 1 -
    e^{-2\lambda/h}.
  \end{align}
\end{pptheorem}

\begin{remark}[Star-like genealogy, scaling and dominant
  alleles]\label{rem:44}
  \begin{enumerate} 
  \item Again, it is crucial to note that only by the scaling of
    $\rho$ by $\alpha/\log\alpha$ a non-trivial limit for the
    structured coalescent arises.
  \item The Theorem must be read as follows: any single line
    recombines out of the sweep (i.e.\ experiences a transition 3.,
    compare with Definition \ref{def:K}) with approximate probability
    $1-e^{-\lambda/h}$. Moreover, any two lines behave independently
    and coalesce if and only if both lines do not recombine. In
    particular, this interpretation shows the star-like approximation
    for the genealogy at the neutral locus.
  \item It is now possible to compare the reduction in diversity of a
    beneficial allele with dominance $h$ and selection coefficient
    $\alpha_h$ and a co-dominant beneficial allele with selection
    coefficient $\alpha_{1/2}$. As the Theorem shows, the reduction in
    diversity only depends on $\rho \log\alpha/(\alpha h)$. Hence,
    ignoring logarithmic terms, the diversity reduction is
    approximately the same if $2h\alpha_h= \alpha_{1/2}$. In other
    words, the effect of a beneficial allele with selection
    coefficient $\alpha$ and dominance $h$ is similar as for a
    co-dominant one with selection coefficient $2h\alpha$.
  \item We can approximate the expected reduction in heterozygosity at
    the end of a sweep in the case $0<h< 1$. Using the same notation
    as in Remark \ref{rem:35}, with $\rho \alpha/\log\alpha \to
    \lambda$ as in the above result, we find that
    \begin{align}
      \label{eq:rea1}
      \mathbb P_0^{\alpha,h}[H_2(T^\ast)] & \approx
      \big(1-e^{-2\lambda/h}\big)\cdot \mathbb P_0^{\alpha,h}[H_2(0)]
	\nonumber\\
      & \approx \big(1-e^{-2\rho \log\alpha/(\alpha h)}\big)\cdot
      \mathbb P_0^{\alpha,h}[H_2(0)].
    \end{align}
  \item The proof uses (and establishes) the fact that for $\alpha \to
    \infty$ recombination and coalescence can only happen while the
    frequency of the beneficial allele is small ($X_t < \epsilon$,
    first phase in Table 1).  Note that the result uses the scaling of
    the recombination rate $\rho \sim \alpha/\log\alpha$. This is in
    stark contrast to the recessive case where $\rho$ must be of order
    $\sqrt\alpha$.
  \item The above theorem does not cover the case $h=1$. For such a
    dominant sweep, recall from Theorem \ref{T3} that the duration of
    the sweep is of the order of $1/\sqrt\alpha$ in the diffusion time
    scale. However, this order is due to the final phase of the
    fixation process, where the frequency of the beneficial allele is
    near $1$. During this final phase, it is possible that additional
    recombination events to the wild-type background occur, plus
    back-recombination events to the beneficial background. In
    particular, (i) in the proof of Theorem \ref{T2h} is not true in
    the case $h=1$.
  \item Our result shows that methods, which rely on the star-like
    approximation (such as \cite{NielsenEtAl2005}), can easily be
    adapted to the case of intermediate dominance. In contrast, they
    are expected to fail for recessive sweeps, i.e.\ $h\alpha$ of
    order 1, even if selection is very strong. Note that the scaling
    result for the star-like approximation uses the limit $\alpha \to
    \infty$ for constant $h$ and should therefore only be applied to
    cases where $h\alpha \gg 1$.
  \end{enumerate}
\end{remark}

\section{Simulations}
\label{S:div}
In this section, we show simulation results indicating the
implications of our findings for data analysis. All simulations were
done using the program {\it msms} \citep{Ewingmsms2010}. This program
is able to simulate samples of arbitrary size of homologous genetic
material of any sequence length for arbitrary selection scenarios on a
single bi-allelic locus, including temporally or spatially
heterogeneous selection and dominance.

We concentrate our simulations on the reduction of sequence diversity
at the neutral locus. In particular, we compare the effects of
selective sweeps for recessive and co-dominant alleles. Since the same
value for the selection strength $\alpha$ will lead to a much weaker
genetic footprint for the recessive case, $h=0$, as compared to
$h=0.5$, we give a heuristic argument in Section \ref{S:51} that
describes which values of $\alpha$ correspond to a footprint with a
given strength for recessive and co-dominant alleles. The resulting
sequence diversity is described in Section \ref{S:52}. A power
analysis of Tajima's $D$, a common test statistic to reject the
standard neutral model, is given Section \ref{S:53}.

\subsection{Comparison of recessive and co-dominant sweeps}
\label{S:51}
The difference of recessive and co-dominant sweeps is best seen from
the approximations of the genealogy at the neutral locus; compare
Remarks \ref{rem:35} and \ref{rem:44}. In particular, we have seen
that a star-like approximation can be used for strong selection and
$h>0$ (see \eqref{eq:rea1}), while the case of $h=0$ is best described
by competing Poisson processes for recombination and coalescence with
constant rates; see (\ref{eq:rea}). Moreover, note that
\eqref{eq:rea1} and \eqref{eq:rea} already give first-order estimates
for the reduction of heterozygosity due to a selective sweep.

We can ask, which selection strength $\alpha_0$ for a recessive allele
is needed to obtain the same expected reduction in heterozygosity as a
co-dominant sweep with a given selection coefficient
$\alpha_{1/2}$. Equating \eqref{eq:rea} and \eqref{eq:rea1},
\begin{align*}
  \frac{2\rho}{\sqrt\alpha_{0} + 2\rho} = 1 -
  e^{-4\rho\log\alpha_{1/2}/\alpha_{1/2}},
\end{align*}
we obtain the following condition for small $\rho$, 
\begin{align}\label{eq:rea3}
  \alpha_0 \approx \Big(\frac{\alpha_{1/2}}{2\log\alpha_{1/2}}\Big)^2.
\end{align}
We note that this relation produces sweeps of almost identical total
length. Indeed, equating the leading order terms for the fixation
times \eqref{eq:T11a} and \eqref{eq:T12} leads to a selection
coefficient for the recessive case of $\tilde{\alpha}_0 \approx 1.07
\alpha_0$, with $\alpha_0$ given in \eqref{eq:rea3}.  In our
simulations, we will use the particular pair $\alpha_{1/2}=1000$ and
$\alpha_{0}=5300,$ which approximately fulfill \eqref{eq:rea3}.

\subsection{The reduction in sequence diversity}
\label{S:52}
A commonly used measure for sequence diversity is the nucleotide
diversity, which we denote by $\widehat\theta_T$. For a sample of $n$
homologous sequences of neutral loci,
$$ \widehat\theta_T = \frac{1}{\binom n 2} \sum_{1\leq i < j \leq n} H_{ij}.$$
Here, $H_{ij}$ is the number of differences between the $i$th and
$j$th locus. Recall that $\widehat\theta_T$ is an unbiased estimator
for the population mutation rate $\theta$ under the standard neutral
model \citep{Tajima1983}.

~

Note that $\mathbb E[\widehat\theta_T] = \mathbb E[H_2]$, where $H_2$
is the heterozygosity in the population. Hence, we can use
\eqref{eq:rea1} and \eqref{eq:rea} for a prediction of
$\widehat\theta_T$ at the end of the selective sweep. We have tested
these predictions using coalescent simulations under the infinite
sites model of a stretch of DNA for a sample of size $n=50$ for
several pairs of selection coefficients according to
\eqref{eq:rea3}. Results for the particular pair $\alpha_{1/2}=1000$
and $\alpha_{0}=5300$ are shown in Figure \ref{fig7}. The match for
the reduction in heterozygosity is good as long as $\rho\ll\alpha$
(compare the slope for $\widehat\theta_T$ near $\rho = 0$ for $h=0$
and $h=0.5$). For larger values of $\rho$, there are several competing
forces and all approximations show deviations. For $h = 0.5$, the
star-like approximation overestimates the length of the genealogy and
gives a too small valley of reduced heterozygosity. For $h = 0$, our
Poisson process approximation assumes that we can neglect fluctuations
of the frequency of the beneficial allele around $1/\sqrt\alpha$. The
simulations show that a too broad valley of reduced heterozygosity is
predicted (see Figure \ref{fig7}(A)). Note that the star-like
approximation for $h=0.5$ and the Poisson process approximation for
$h=0$, give errors of comparable magnitude.



~

\subsection{The power of Tajima's  $D$ }
\label{S:53}
Another unbiased estimator for the population mutation rate $\theta$
is given by Watterson's $\theta$,
$$\widehat \theta_W = \frac{S}{\sum_{i=1}^{n-1} \tfrac 1i},$$
where $S$ is the total number of single nucleotide polymorphisms
(SNPs) in all $n$ lines \citep{Watterson1975}. In order to reject the
neutral model, we consider the frequently used statistics Tajima's $D$
\citep{Tajima1989}. This statistics was one of the first test
statistics for neutrality tests and is proportional to $\widehat\theta_T -
\widehat \theta_W$. It is a classical result that Tajima's $D$ is
negative for selective sweeps.

The simulations show a more shallow increase for $\widehat \theta_W$
in the recessive case (see Figure \ref{fig7}(A)). This implies a
smaller difference of $\widehat \theta_W$ and $\widehat\theta_T$ and
hence a less negative value for Tajima's $D$ relative to sweep with
$h=0.5$, as seen in Figure \ref{fig7}(B). Also, this finding is easily
understood from the approximated coalescent histories. With a
star-like genealogy, only single lines of descent recombine out of the
sweep, leading to an excess of singletons (more generally: low
frequency polymorphisms) relative to the neutral expectation. This
bias in the frequency spectrum is indicated by a shift to negative
values of Tajima's $D$. In contrast, coalescence events will often
occur prior to recombination events under a competing Poisson process
scheme, valid for recessive sweeps. As a consequence, lines of descent
that recombine out of the sweep will often have multiple descendants
among the sequences in the sample. The shift towards low frequency
polymorphisms is therefore less pronounced and Tajima's $D$ less
negative. Similar results have previously been reported for moderately
recessive sweeps ($h=0.1$) by \cite{pmid16219788}.

The approximate genealogy with constant coalescence and recombination
rates is equivalent to a model of a structured population with
constant migration among two islands representing the beneficial and
the ancestral background. Depending on the migration rate (and the
sampling scheme) both a surplus ($D<0$) and a deficit ($D>0$) of low
frequency polymorphisms can result under this scenario. For the
recessive sweep, we see that the expected Tajima's $D$ indeed turns
positive at a larger recombination distance to the selected site.

Since the deviation of Tajima's $D$ from zero (in either direction) is
relatively smaller for recessive sweeps, we expect that also the power
to detect selection in a simple test based on $D$ is reduced. This is
confirmed by the power table in Figure \ref{fig8}, consistent with
previous results on partially recessive sweeps by \cite{pmid16687733}.

\begin{figure}
  \centering \hspace*{0cm}(A) \hspace{.45\textwidth} (B)\\ 
\fboxsep0pt
{\includegraphics[width=.48\textwidth,trim=0 0 34 34,clip]{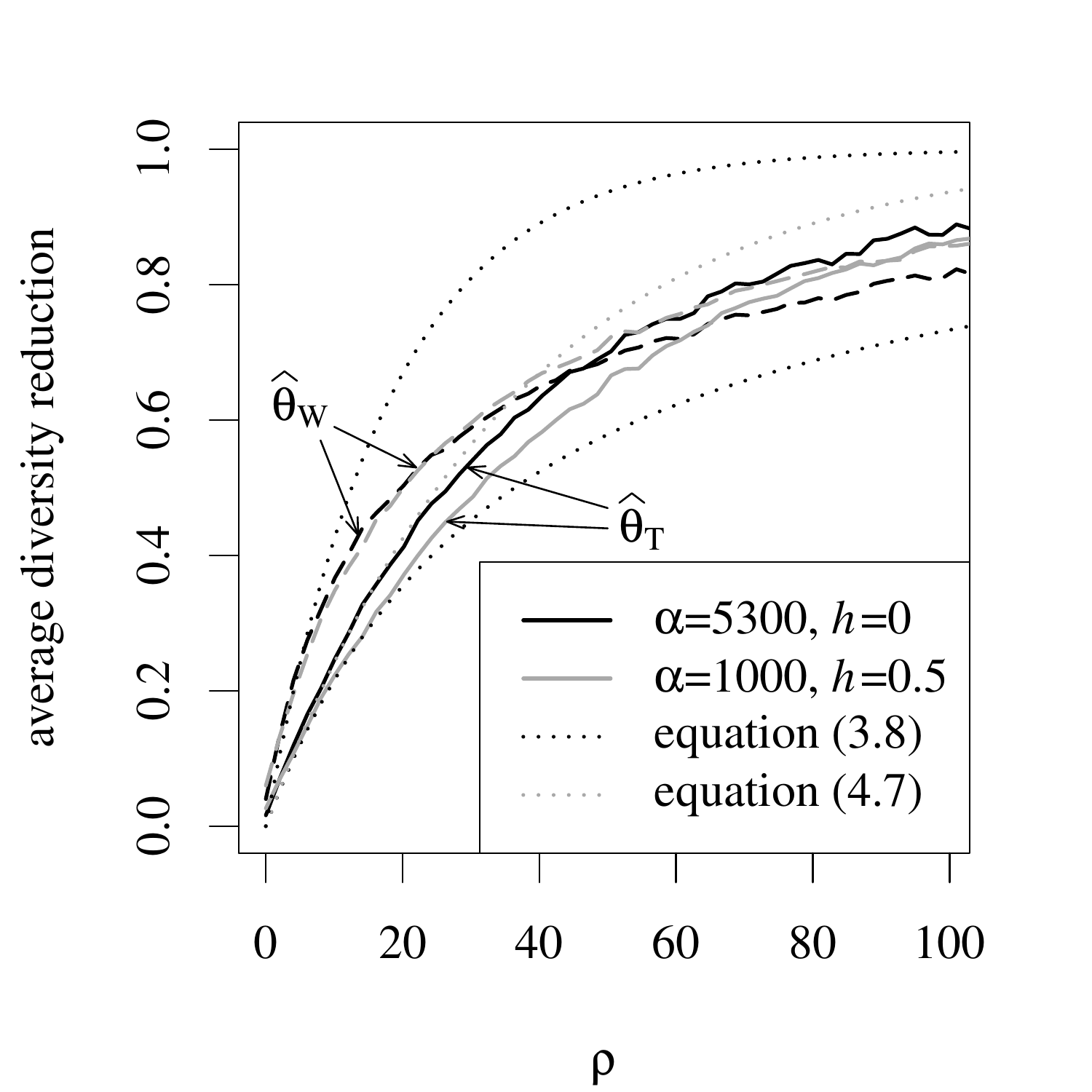}}\hfill
{\includegraphics[width=.48\textwidth,trim=0 0 34 34,clip]{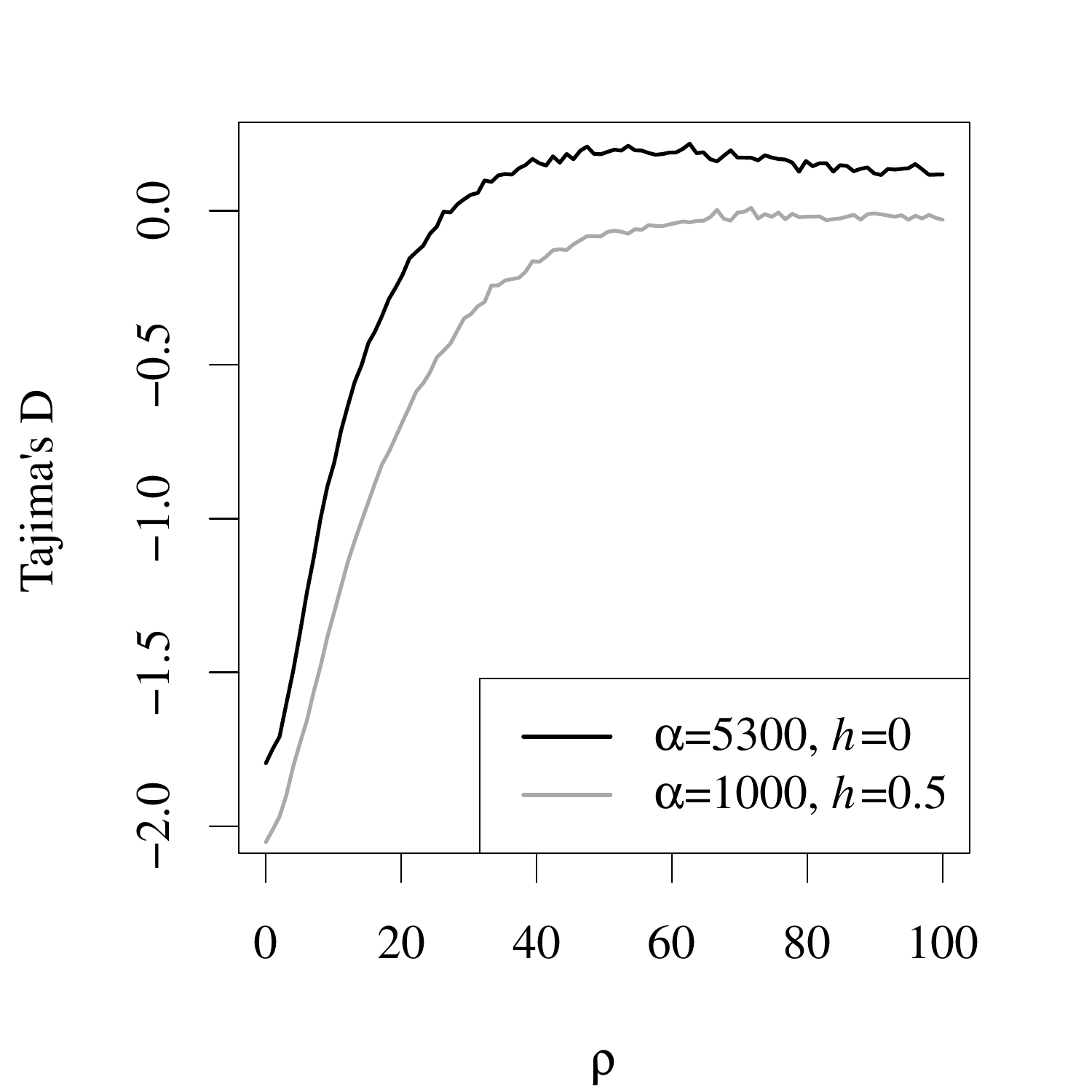}}
  \caption{\label{fig7}%
(A) The average reduction of nucleotide
    diversity $\widehat\theta_T$ and Watterson's estimator
    $\widehat\theta_W$ for the population mutation rate $\theta$ after
    a sweep. (B) Tajima's $D$ as a function of the recombination
    distance, $\rho$, for co-dominant and recessive sweeps. }
\end{figure}

\begin{figure}
  \centering \hspace*{0cm}(A) \hspace{.45\textwidth} (B)\\

\fboxsep0pt 
{\includegraphics[width=.48\textwidth,trim=0 0 20 15,clip]{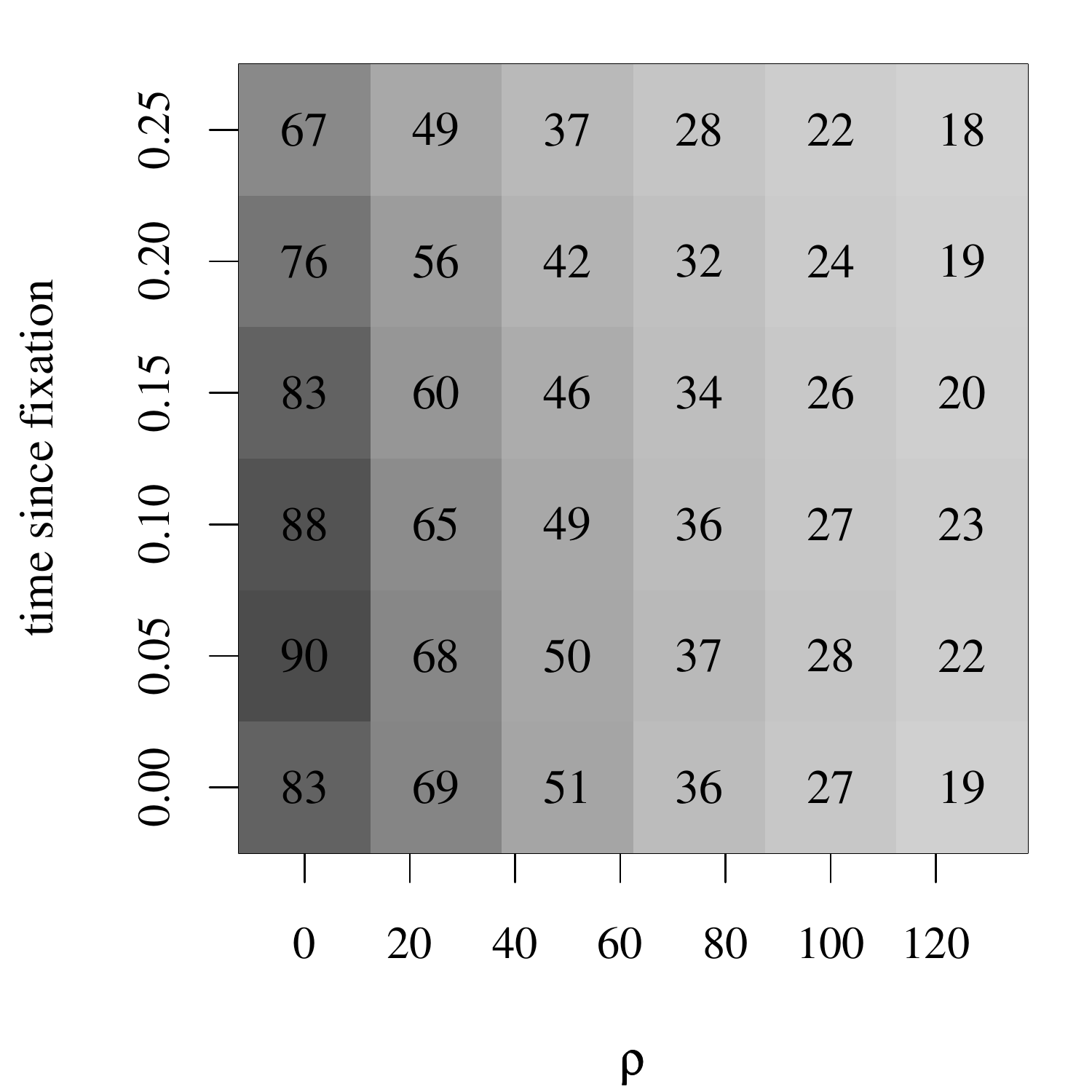}}%
\hfill%
{\includegraphics[width=.48\textwidth,trim=0 0 20 15,clip]{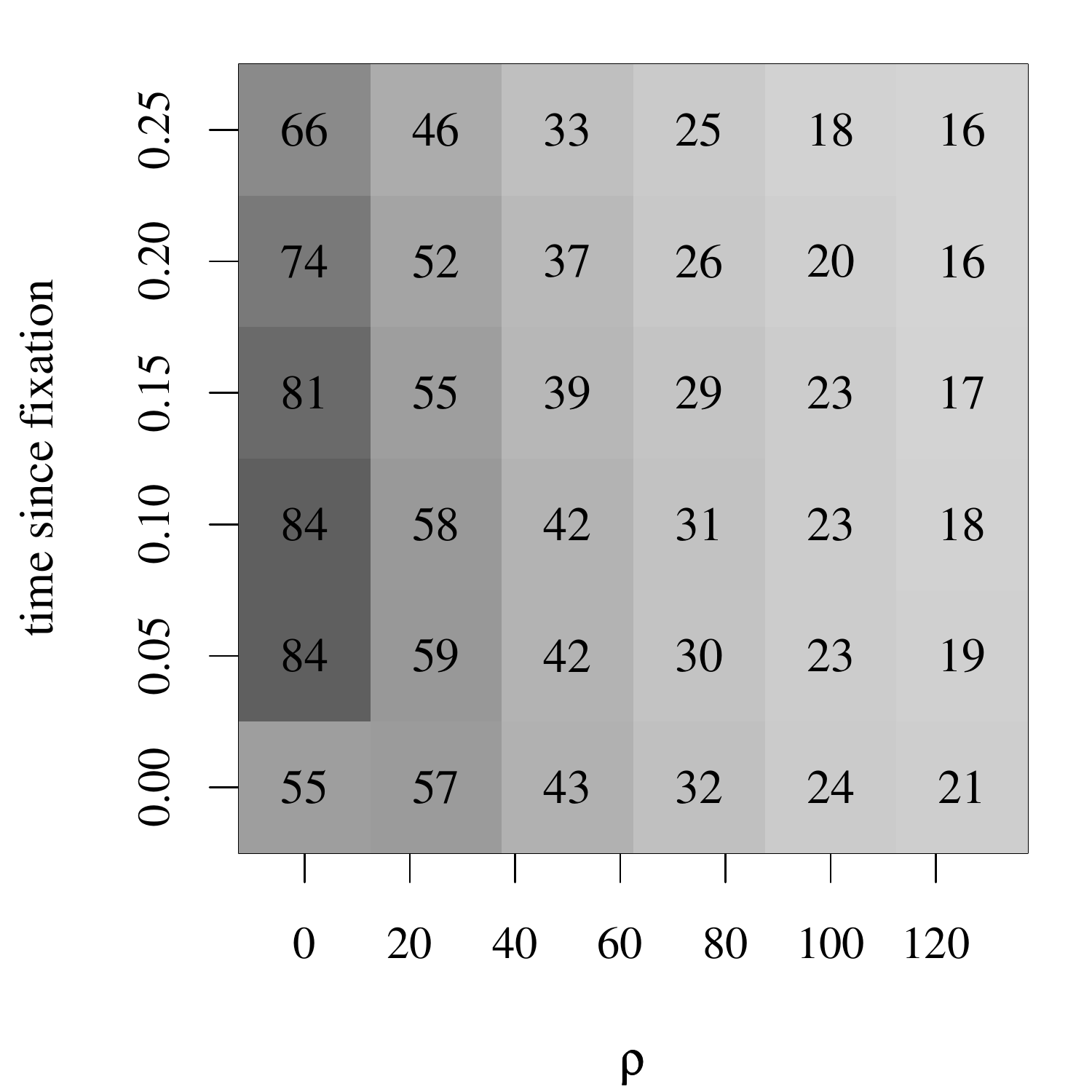}}
\caption{\label{fig8}%
The power of Tajima's $D$, depending on
    dominance, recombination distance and time since completion of the
    selective sweep. (A) co-dominant and (B) recessive, parameters as
    in Figure \ref{fig7}(B).
}
\end{figure}

\section{Proofs}
\label{S:proofs}
All our assertions are dealing with the case $\alpha\to\infty$. Hence,
all '$\approx$' in our proofs are to be read as
$\stackrel{\alpha\to\infty}\approx$. Frequently, we make use of a
sequence $(c_\alpha)_{\alpha>0}$ of real numbers with
$c_\alpha\to\infty$ slowly with $\alpha\to\infty$. For example, we
write
$$ \int_0^\alpha \frac{1-e^{-\xi}}{\xi}d\xi \approx \int_{c_\alpha}^\alpha 
\frac{1-e^{-\xi}}{\xi}d\xi \approx \int_{c_\alpha}^\alpha
\frac{1}{\xi}d\xi \approx \log\alpha.$$

Our proofs are based on classical one-dimensional diffusion theory;
see e.g.\ \cite{KarlinTaylor1981, Ewens2004}. The solutions of the SDE
as given in \eqref{eq:SDE0} have infinitesimal mean and variance
$$ \mu(x) = \alpha (h + x(1-2h))x(1-x), \qquad \sigma^2(x) = x(1-x).$$
The scale function is
\begin{align}\label{eq:scale}
  S(x) & = \int_0^x \exp\Big( - \int_0^\xi
  \frac{2\mu(\eta)}{\sigma^2(\eta)} d\eta\Big)d\xi
\nonumber \\
& = \int_0^x
  \exp\big( -2h\alpha \xi - \alpha \xi^2(1-2h)\big) d\xi.
\end{align}
Since $0$ and $1$ are exit boundaries for the diffusion, and for $T_0$
and $T_1$as defined in \eqref{eq:stop},
\begin{align}\label{eq:fix1}
  \mathbb P_x^{\alpha,h}[T_1<T_0] = \frac{S(x)}{S(1)} =
  \frac{\displaystyle\int_0^x e^{-2h\alpha
      \eta-\alpha\eta^2(1-2h)}d\eta}{\displaystyle\int_0^1
    e^{-2h\alpha \eta-\alpha\eta^2(1-2h)}d\eta}.
\end{align}
The Green function for $X^\ast$, the diffusion which is conditioned to
hit 1 and started in $x$, is given by
\begin{eqnarray}
  \label{eq:greenX}
    G^\ast(x,\xi) & =&
    \begin{cases} \frac{\displaystyle 2(S(1)-S(\xi))}{\displaystyle
        S(1)}\frac{\displaystyle
        S(\xi)}{\displaystyle \sigma^2(\xi)S'(\xi)}, & x\leq \xi, \\[4ex]
      \frac{\displaystyle 2(S(1) - S(x))S(\xi)}{\displaystyle
        S(1)}\frac{\displaystyle S(\xi)}{\displaystyle \sigma^2(\xi)
         S'(\xi)S(x)}, & x> \xi \end{cases}
 \nonumber\\[6ex]
& =&
    \begin{cases}
	2 \frac{\displaystyle\int_0^\xi e^{-2h\alpha
          \eta-\alpha\eta^2(1-2h)}d\eta\int_\xi^1
        e^{-2h\alpha\eta-\alpha\eta^2(1-2h)}d\eta}{\displaystyle
        \xi(1-\xi) e^{-2h\alpha\xi-\alpha\xi^2(1-2h)} \int_0^1
        e^{-2h\alpha\eta-\alpha\eta^2(1-2h)}d\eta},
	& x\leq \xi,\qquad
	\\[7ex]
      2 \frac{\displaystyle \int_x^1 e^{-2h\alpha
          \eta-\alpha\eta^2(1-2h)}d\eta }{\displaystyle
        \xi(1-\xi) e^{-2h\alpha\xi-\alpha\xi^2(1-2h)}\int_0^1
        e^{-2h\alpha\eta-\alpha\eta^2(1-2h)}d\eta} 
	\\
	\qquad \cdot\frac{\displaystyle \Big( \int_0^\xi
        e^{-2h\alpha\eta-\alpha\eta^2(1-2h)}d\eta\Big)^2}{\displaystyle
 	\int_0^x e^{-2h\alpha\eta-\alpha\eta^2(1-2h)}d\eta},
	& x> \xi.
	\end{cases}
\end{eqnarray}
Here, $G^\ast(x,\xi)d\xi$ is the average time the conditioned
diffusion spends in \linebreak  $(\xi,\xi+d\xi)$ before hitting 1. In particular,
we will use that
\begin{align}
  \label{eq:green2}
  \mathbb E_0[T^\ast] = \int_0^1 G^\ast(0,\xi)d\xi, \qquad \mathbb
  V_0[T^\ast] = 2\int_0^1 \int_0^\xi G^\ast(0,\xi) G^\ast(\xi,\eta)
  d\eta d\xi.
\end{align}
While the first identity is a classical result in
diffusion theory, the second identity can e.g.\ be read off from
\citet[equation 2.1.1]{DawsonGorostizaWakolbinger2001}.

\subsection{Proofs from Section \ref{S:results}}
\label{S:proofs1}
In this Section, we fix $h=0$. We will start with our key Lemma
\ref{lkey}, which states that a scaling of the diffusion $\mathcal X$
immediately gives our result that the sweep length is of order
$1/\sqrt\alpha$. Afterwards we provide proofs for Proposition
\ref{P1}, and Theorems \ref{T1} and \ref{T2}.

\subsection*{A key lemma}
In order to study the diffusion \eqref{eq:SDE} and the structured
coalescent, we use a rescaling argument. We need a definition first.

\begin{definition}\label{def:Y}
  The diffusion $\mathcal Y = (\mathcal Y_\tau)_{\tau\geq 0}$, which
  takes values in $[0, \infty]$, is the unique solution of the SDE
  $$ dY = Y^2 d\tau + \sqrt{Y}dW,$$
  for which $Y_\tau = \infty$ is a trap. We set
$$
T_{\mathcal Y} := \inf\{\tau\geq 0: Y_\tau = \infty\}.
$$
The process $\mathcal Y^\ast = (Y_\tau^\ast)_{\tau\geq 0}$ is given
  as $\mathcal Y$, conditioned on the event $\{T_{\mathcal
    Y}<\infty\}$. We write $T_{\mathcal Y}^\ast$ for $T_{\mathcal Y}$,
  also conditioned on $\{T_{\mathcal Y}<\infty\}$.
\end{definition}

\begin{remark}[The diffusion \boldmath $\mathcal Y^\ast$]%
\label{remY}
  The diffusion $\mathcal Y^\ast$ will play a crucial role in our
  proofs. Note that the scale function of this diffusion is
  \begin{align*}
    S^{\mathcal Y}(x) = \int_0^x e^{-\xi^2} d\xi.
  \end{align*}
  For the Green function of $\mathcal Y^\ast$, 
  \begin{align}\label{eq:greenY}
    G^{\mathcal Y, \ast}(x,\xi) = \begin{cases} 
      \displaystyle\frac{4}{\sqrt\pi}\frac{\displaystyle\int_\xi^\infty
        e^{-\eta^2}d\eta  \int_0^\xi e^{-\eta^2}
        d\eta}{\displaystyle\xi  e^{-\xi^2}}, & x\leq
      \xi,\\[5ex]
      \displaystyle\frac{4}{\sqrt\pi}\frac{\displaystyle\int_x^\infty
        e^{-\eta^2}d\eta  \Big(\int_0^\xi e^{-\eta^2}
        d\eta\Big)^2}{\displaystyle\xi  e^{-\xi^2} \int_0^x
        e^{-\eta^2} d\eta}, & x> \xi. \end{cases}
  \end{align}
\end{remark}

~

\begin{lemma}[Key lemma]\label{lkey}
  Let $T^\ast, \mathcal X^{\ast}$ be as in Definition \ref{def:X} and
  $T_{\mathcal Y}, T_{\mathcal Y}^\ast, \mathcal Y^\ast$ be as in
  Definition \ref{def:Y}, such that $\mathcal X^{\ast}$ and $\mathcal
  Y^\ast$ are started in $0$. Then, for $\widetilde{\mathcal X}^{\ast}
  := (\widetilde X_\tau^{\ast})_{\tau\geq 0}$, $ \widetilde
  X_\tau^{\ast}:= \sqrt\alpha X^{\ast}_{\tau/\sqrt\alpha}$,
  \begin{align}
    \label{eq:key1}
    \widetilde{\mathcal X}^{\ast} \Rightarrow \mathcal Y^\ast \;
    \text{ as $\alpha\to\infty$}.
  \end{align}
  Moreover, 
  \begin{align}
    \label{eq:key2}
    \mathbb E_0^{\alpha,h=0}[T^\ast] - \tfrac{1}{\sqrt\alpha}\mathbb
    E_0[T_{\mathcal Y}^\ast]
    &\stackrel{\alpha\to\infty}\approx \frac{3\log\alpha}{2\alpha},\\
    \label{eq:key3}
    \mathbb E_0[T_{\mathcal Y}^\ast] & \stackrel{\phantom{\alpha\to\infty}}{=}
    \frac{4}{\sqrt\pi}c_{\text{cat}},
  \end{align}
  where $c_{\text{cat}} \approx 0.916$ is Catalan's constant. In
  particular,
  \begin{align}
    \label{eq:key4}
    \sqrt\alpha \cdot T^{\ast} \Rightarrow T_{\mathcal Y}^\ast \qquad
    \text{ as } \alpha\to\infty.
  \end{align}
\end{lemma}

\begin{proof}
  For \eqref{eq:key1}, we start by a change of variables
  $$ \widetilde X_t = \sqrt{\alpha} X_{t/\sqrt{\alpha}}.$$
  By changing the time scale to $d\tau = \sqrt\alpha dt$, we obtain by
  It\^o's formula that
  \begin{align*}
    d\widetilde X & = \sqrt{\alpha} dX = \sqrt\alpha \widetilde X^2(1
    - \widetilde X/\sqrt\alpha)dt + \sqrt{\alpha \tfrac{\widetilde
        X}{\sqrt\alpha}(1-\widetilde X/\sqrt\alpha)}dW \\ & =
    \widetilde X^2(1-\widetilde X/\sqrt\alpha)d\tau + \sqrt{\widetilde
      X(1-\widetilde X/\sqrt\alpha)}dW.
  \end{align*}
  Since $\alpha\to\infty$, we see that
  $$\widetilde{\mathcal X} \Rightarrow \mathcal Y \text{ as $\alpha\to\infty$}$$
  which also implies \eqref{eq:key1}.

  For \eqref{eq:key3}, with a little help from {\sc Mathematica} to
  evaluate the last integral,
  \begin{align}\label{eq:S1}
    \mathbb E_0[T_{\mathcal Y}^\ast] = \int_0^\infty G^{\mathcal Y,
      \ast}(0,\xi) d\xi = \frac{4}{\sqrt\pi} \int_0^\infty
    \frac{\int_\xi^\infty e^{-\eta^2}d\eta \int_0^\xi
      e^{-\eta^2}d\eta}{\xi e^{-\xi^2}} d\xi = \frac{4
      c_{cat}}{\sqrt{\pi}}.
  \end{align}
For \eqref{eq:key2}, recall the Green function for the diffusion
  $\mathcal X^\ast$ from \eqref{eq:greenX}. 
  We use 
\begin{align*}
 \int_\xi^\infty e^{-\eta^2}
&d\eta \stackrel{\xi\to\infty}\approx 
  \frac{e^{-\xi^2}}{2\xi}, \qquad
\int_0^\infty e^{-\eta^2}d\eta =
  \frac{\sqrt\pi}{2},
\\
& \frac{1}{\xi(1-\xi/\sqrt{\alpha})} =
  \frac{1}{\xi} + \frac{1}{\sqrt\alpha - \xi}
\end{align*}
and \eqref{eq:S1} to obtain
  \begin{align*}
    \mathbb E_0^{\alpha, h=0}[T^\ast] & = \int_0^1 G^{\mathcal X,
      \ast}(0,\xi)d\xi \approx \frac{4}{\sqrt{\pi\alpha}}
    \int_0^{\sqrt{\alpha}} \frac{ \int_\xi^{\sqrt{\alpha}} e^{-\eta^2}
      d\eta \int_0^\xi e^{-\eta^2}
      d\eta}{\xi(1-\xi/\sqrt\alpha)e^{-\xi^2}}d\xi \\ & \approx
    \tfrac{1}{\sqrt\alpha}\mathbb E_0[T_{\mathcal Y}^\ast] +
    \tfrac{4}{\sqrt{\pi\alpha}}\big( \Delta_1 + \Delta_2 + \Delta_3 +
    \Delta_4\big)
  \end{align*}
  with (recall the sequence $c_\alpha$ going to infinity slowly
  enough)
  \begin{align*}
    \Delta_1 & := - \int_{\sqrt\alpha}^\infty \frac{\int_\xi^\infty
      e^{-\eta^2} d\eta \int_0^\xi e^{-\eta^2} d\eta}{\xi
      e^{-\xi^2}}d\xi\approx - \tfrac{\sqrt\pi}{4}
    \int_{\sqrt\alpha}^\infty \frac{1}{\xi^2} d\xi =
    - \tfrac{\sqrt\pi}{4\sqrt\alpha},
\\
    \Delta_2 & := - \int_0^{\sqrt\alpha}
    \frac{\int_{\sqrt\alpha}^\infty e^{-\eta^2} d\eta \int_0^\xi
      e^{-\eta^2} d\eta}{\xi e^{-\xi^2}}d\xi \approx
    -\tfrac{1}{2\sqrt\alpha} \int_{0}^{\sqrt\alpha} \frac{e^{\xi^2 -
        \alpha} \int_0^\xi e^{-\eta^2} d\eta}{\xi} d\xi
\\
& {}\kern.8ex\approx
    -\tfrac{1}{2\sqrt\alpha} \int_{\sqrt\alpha /2}^{\sqrt\alpha}
    \frac{e^{\xi^2 - \alpha} \int_0^\xi e^{-\eta^2} d\eta}{\xi} d\xi
    \approx -\tfrac{\sqrt\pi}{4\sqrt{\alpha}} e^{-\alpha}
    \int_{\sqrt\alpha/2}^{\sqrt\alpha} \frac{e^{\xi^2}}{\xi} d\xi
    \approx -\frac{\sqrt\pi}{8\alpha^{3/2}},
\\
    \Delta_3 & := \int_0^{\sqrt\alpha-c_\alpha}
    \frac{\int_\xi^{\sqrt\alpha} e^{-\eta^2} d\eta \int_0^\xi
      e^{-\eta^2} d\eta}{(\sqrt\alpha - \xi) e^{-\xi^2}} d\xi \approx
    \int_{c_\alpha}^{\sqrt\alpha-c_\alpha} \frac{\int_\xi^{\infty}
      e^{-\eta^2}d\eta\int_0^\xi e^{-\eta^2}d\eta}{(\sqrt{\alpha} -
      \xi)e^{-\xi^2}} d\xi
\\
& {}\kern.8ex\approx \tfrac{\sqrt\pi}{4}
    \int_{c_\alpha}^{\sqrt\alpha-c_\alpha} \frac{1}{\xi(\sqrt\alpha -
      \xi)} d\xi = \tfrac{\sqrt\pi}{4\sqrt\alpha}
    \int_{c_\alpha}^{\sqrt\alpha-c_\alpha} \frac 1\xi +
    \frac{1}{\sqrt\alpha - \xi}d\xi
\\
& {}\kern.8ex\approx
    \tfrac{\sqrt\pi}{2\sqrt\alpha}\log(\sqrt\alpha) =
    \tfrac{\sqrt\pi}{4\sqrt\alpha}\log\alpha,\\
    \Delta_4 &:= \int_{\sqrt\alpha-c_\alpha}^{\sqrt\alpha}
    \frac{\int_\xi^{\sqrt\alpha} e^{-\eta^2} d\eta \int_0^\xi
      e^{-\eta^2} d\eta}{(\sqrt\alpha - \xi) e^{-\xi^2}} d\xi \approx
    \tfrac{\sqrt\pi}{2} \int_{\sqrt\alpha-c_\alpha}^{\sqrt\alpha}
    \frac{\int_\xi^{\sqrt\alpha}
      e^{-(\eta-\xi)(\eta+\xi)}d\eta}{\sqrt\alpha - \xi}d\xi
\\
&{}\kern.8ex \approx \tfrac{\sqrt\pi}{2}
    \int_{\sqrt\alpha-c_\alpha}^{\sqrt\alpha}
    \frac{\int_\xi^{\sqrt\alpha} e^{-(\eta-\xi)2\sqrt\alpha}
      d\eta}{\sqrt\alpha - \xi}d\xi = \tfrac{\sqrt\pi}{4 \sqrt\alpha}
    \int_{\sqrt\alpha-c_\alpha}^{\sqrt\alpha} \frac{1-e^{-(\sqrt\alpha
        - \xi)2\sqrt\alpha}}{\sqrt\alpha - \xi} d\xi
\\
& {}\kern.8ex =
    \tfrac{\sqrt\pi}{4\sqrt\alpha} \int_0^{c_\alpha} \frac{1-e^{-\xi
        2\sqrt\alpha}}{\xi} d\xi = \tfrac{\sqrt\pi}{4\sqrt\alpha}
    \int_0^{\sqrt\alpha c_\alpha} \frac{1-e^{-2\xi}}{\xi} d\xi \approx
    \tfrac{\sqrt\pi}{8\sqrt\alpha}\log\alpha.
  \end{align*}
  The result \eqref{eq:key2} follows.
\end{proof}

\noindent%
\textbf{Proof of Proposition \ref{P1}}\\
Due to Lemma \ref{lkey}, we can compute, by classical diffusion
theory,
\begin{align*}
  \mathbb P_{\varepsilon_\alpha}^{\alpha, h}[T_1<T_0] & \approx
  \mathbb P_{\sqrt\alpha\cdot \varepsilon_\alpha}[T_{\mathcal Y}
  <\infty] = \frac{S^{\mathcal Y}(\sqrt\alpha\cdot
    \varepsilon_\alpha)}
  {S^{\mathcal Y}(\infty)} \\
  & = \frac{\int_0^{\sqrt\alpha \cdot \varepsilon_\alpha} e^{-\eta^2}
    d\eta}{\int_0^\infty e^{-\eta^2} d\eta} \approx \sqrt\alpha \cdot
  \varepsilon_\alpha \cdot\frac{2}{\sqrt\pi}
\end{align*}
and we are done.

\bigskip
\noindent%
\textbf{Proof of Theorem \ref{T1}}\\
Clearly, \eqref{eq:T11a} is a combination of \eqref{eq:key2} and
\eqref{eq:key3}. For \eqref{eq:T12b} we use
\eqref{eq:green2} and write 
\begin{eqnarray}
  \label{eq:var}
    \alpha \mathbb V_0^{\alpha, h=0}[T^\ast] & =& 8 \alpha \int_0^1
    \int_0^\xi\frac{\Big(\int_\xi^1 e^{-\alpha \zeta^2}d\zeta\Big)^2
      \Big(\int_0^\eta e^{-\alpha \zeta^2}d\zeta\Big)^2}{\Big(\int_0^1
      e^{-\alpha \zeta^2}d\zeta\Big)^2\xi(1-\xi)\eta(1-\eta)
      e^{-\alpha\xi^2} e^{-\alpha\eta^2}} d\eta d\xi \enskip
\nonumber \\
& \approx&
    8\int_0^\infty \int_0^\xi\frac{\Big(\int_\xi^\infty
      e^{-\zeta^2}d\zeta\Big)^2 \Big(\int_0^\eta
      e^{-\zeta^2}d\zeta\Big)^2}{\Big(\int_0^\infty
      e^{-\zeta^2}d\zeta\Big)^2\xi\eta e^{-\xi^2} e^{-\eta^2} } d\eta
    d\xi < \infty.
\end{eqnarray}

\bigskip
\noindent%
\textbf{Proof of Theorem \ref{T2}}\\
We use the same notation as in Lemma \ref{lkey} (in particular
$\sqrt\alpha dt = d\tau$) and transform the structured coalescent to a
structured coalescent conditioned on the process $\mathcal Y$. For
\eqref{T22}, note that the coalescence rate at time $t$ is given by
$$\frac 1{X_t} dt = \frac{\sqrt\alpha}{\sqrt\alpha \widetilde X_\tau} d\tau \Rightarrow 
\frac{1}{Y_\tau} d\tau \text{ as $\alpha\to\infty$.}$$ Recombinations
occur at rate
\begin{align}
  \label{eq:rec10}
  \rho(1-X_t) dt = \frac{\rho}{\sqrt\alpha}\Big( 1 - \frac{\widetilde
    X_\tau}{\sqrt\alpha}\Big)d\tau \approx \lambda.
\end{align}
For back recombination from the wild-type to the beneficial
background,
$$\rho X_t dt =  \frac{\rho}{\sqrt\alpha} \frac{\widetilde X_\tau}{\sqrt\alpha}d\tau
\approx 0.$$ All three limits are valid for all $0\leq \tau <
T_{\mathcal Y}^\ast$.

So, for the structured coalescent conditioned on $\mathcal Y$, we find
that the coalescence rate is $1/Y_\tau d\tau$, and the recombination
rate is $\lambda$. During times when $Y_\tau$ is of order 1, we see
that both, coalescence and recombination happen at rate of order 1. In
particular, we have shown \eqref{T22}. For \eqref{T21}, since back
recombination can be ignored, by \eqref{eq:rec10},
\begin{align*}
  \mathbb P_0^{\alpha, h=0}[\mathcal K_{T^\ast}=(1,0)|\mathcal K_0=(1,0)]
  & = \mathbb E^{\alpha, h=0}_0[e^{-\lambda (\sqrt\alpha \cdot
    T^\ast)}].
\end{align*}

\subsection{Proofs from Section \ref{S:results2}}
\label{S:proofs2}
We directly use \eqref{eq:fix1} in order to prove Proposition
\ref{P2}, and \eqref{eq:green2} in order to prove Theorem \ref{T3}.
We will use several approximation results on integrals appearing in
these equations. We will use
\begin{equation}\label{eq:app2}
  \begin{aligned}
    \int_0^\alpha e^{\xi^2} d\xi - \frac{e^{\alpha^2}}{2\alpha} & =
    e^{\alpha^2} \Big(\int_0^\alpha e^{-2\xi\alpha}(e^{\xi^2}-1)d\xi -
    \int_\alpha^\infty e^{-2\xi\alpha}d\xi\Big) \\ & \approx
    e^{\alpha^2}\int_0^\alpha \xi^2 e^{-2\alpha \xi} d\xi \approx
    \frac{e^{\alpha^2}}{4\alpha^3},
  \end{aligned}
\end{equation}
as well as
\begin{eqnarray}\label{eq:app3}
  \lefteqn{
    \int_0^\xi e^{-2h\eta + (2h-1)\eta^2/\alpha} d\eta - \frac 1{2h}(1-e^{-2h\xi})}
\quad \nonumber \\
& =& \int_0^\xi
    e^{-2h\eta}(e^{(2h-1)\eta^2/\alpha} -1)d\eta
\nonumber\\
& \approx&
    \frac{2h-1}{\alpha}\int_0^\xi e^{-2h\eta}\eta^2 d\eta \approx
    \frac{c(1-e^{-2h\xi}) + de^{-2h\xi}}{\alpha}
\end{eqnarray}
for $\xi\leq\alpha$ and fixed $0<h\leq 1$, for some $0\leq
c$, $d<\infty$. The last equation implies for $\xi=\alpha$
\begin{equation}\label{eq:app1}
  \begin{aligned}
    \int_0^\alpha e^{-2h\eta + (2h-1)\eta^2/\alpha} d\eta - \frac
    1{2h} & \approx\frac{2h-1}{4h^3\alpha}.
  \end{aligned}
\end{equation}

\bigskip\noindent%
\textbf{Proof of Proposition \ref{P2}}\\
For $0<h\leq 1$, we write, using \eqref{eq:fix1} and \eqref{eq:app1},
\begin{align*}
  \mathbb P^{\alpha, h}_{\varepsilon_\alpha}[T_1<T_0] & =
  \frac{\displaystyle \int_0^{\alpha \varepsilon_\alpha} e^{-2h \eta +
      (2h-1)\eta^2/\alpha}d\eta}{\displaystyle \int_0^{\alpha} e^{-2h
      \eta + (2h-1)\eta^2/\alpha}d\eta} \approx
  2h\alpha\varepsilon_\alpha.
\end{align*}

\bigskip\noindent%
\textbf{Proof of Theorem \ref{T3}}\\
1. For $0<h<1$, we write with \eqref{eq:app1}
\begin{eqnarray}
  \label{eq:T3p1}
    \mathbb E_0^{\alpha, h}[T^\ast] & =&  2\int_0^\alpha
    \frac{\displaystyle\int_0^\xi e^{-2h
        \eta+(2h-1)\eta^2/\alpha}d\eta\int_\xi^\alpha
      e^{-2h\eta+(2h-1)\eta^2/\alpha}d\eta}{\displaystyle
      \xi(\alpha-\xi) e^{-2h\xi+(2h-1)\xi^2/\alpha} \int_0^\alpha
      e^{-2h\eta+(2h-1)\eta^2/\alpha}d\eta}d\xi
\nonumber \\
&\approx&
    \frac{4h}{\alpha}\big( A_1 + A_2\big)
\end{eqnarray}
with
\begin{align*}
  A_1 & := \int_0^\alpha \frac{\displaystyle\int_0^\xi e^{-2h
      \eta+(2h-1)\eta^2/\alpha}d\eta\int_\xi^\alpha
    e^{-2h\eta+(2h-1)\eta^2/\alpha}d\eta}{\displaystyle \xi
    e^{-2h\xi+(2h-1)\xi^2/\alpha} },\\ A_2 & := \int_0^\alpha
  \frac{\displaystyle\int_0^\xi e^{-2h
      \eta+(2h-1)\eta^2/\alpha}d\eta\int_\xi^\alpha
    e^{-2h\eta+(2h-1)\eta^2/\alpha}d\eta}{\displaystyle (\alpha-\xi)
    e^{-2h\xi+(2h-1)\xi^2/\alpha}}.
\end{align*}
We use
\begin{align*}
  \int_0^\xi \frac{1-e^{-\eta}}{\eta}d\eta - \log\xi \approx \gamma
\end{align*}
for $\xi\to\infty$ and, if $0\leq \xi\leq\alpha$,
\begin{eqnarray}
  \label{eq:app4}
 \lefteqn{\frac{\int_\xi^\alpha e^{-2h\eta +
        (2h-1)\eta^2/\alpha}d\eta}{e^{-2h\xi + (2h-1)\xi^2/\alpha}}
     -\frac{1 - e^{-(2h - (2h-1)2\xi/\alpha)(\alpha-\xi)}}{2h -
      2(2h-1)\xi/\alpha} }
\quad \nonumber \\
& =& \int_0^{\alpha-\xi} e^{-2h\eta +
      (2h-1)(2\eta\xi+\eta^2)/\alpha}d\eta -
    \int_0^{\alpha-\xi} e^{-(2h - (2h-1)2\xi/\alpha)\eta}d\eta
\nonumber\\
    & =& \int_0^{\alpha-\xi}e^{-(2h-2(2h-1)\xi/\alpha)\eta} \big(
    e^{(2h-1)\eta^2/\alpha}-1\big)d\eta
\nonumber\\
& \approx&
    \frac{2h-1}{\alpha}\int_0^{\alpha-\xi}e^{-(2h-2(2h-1)\xi/\alpha)\eta}
    \eta^2d\eta
\nonumber\\
& =&
    \frac{c(1-e^{-(2h-2(2h-1)\xi/\alpha)(\alpha-\xi)}) + d
      e^{-(2h-2(2h-1)\xi/\alpha)(\alpha-\xi)}}{\alpha}
\end{eqnarray}
for some $c,d\geq 0$, not depending on $\alpha$, in order to obtain
\begin{align}
  4hA_1 \notag & - \frac{1}{h}(\log(2h\alpha)+\gamma)
\\ & \approx
  4h\int_0^\alpha \frac{\int_0^\xi e^{-2h\eta +
      (2h-1)\eta^2/\alpha}d\eta \int_\xi^\alpha e^{-2h\eta +
      (2h-1)\eta^2/\alpha}d\eta}{\xi
    e^{-2h\xi + (2h-1)\xi^2/\alpha}}d\xi
\notag \\
& \qquad {} - \frac{1}{h}\int_0^{2h \alpha}\frac{1-e^{-\xi}}{\xi} d\xi\notag \\
  & \label{eq:app6}\approx \frac{1}{h} \Big( \int_0^\alpha
  \frac{1-e^{-2h\xi}}{\xi} \frac{2h}{2h - 2(2h-1)\xi/\alpha} d\xi -
  \int_0^\alpha \frac{1-e^{-2h\xi}}{\xi}d\xi\Big) \\ & =
  \frac{1}{h}\int_0^\alpha (1-e^{-2h\xi}) \frac{2(2h-1)/\alpha}{2h -
    2(2h-1)\xi/\alpha} d\xi \notag \\ & \approx \frac{1}{h} \log\big(
  2h -
  2(2h-1)\xi/\alpha\big)\Big|_{\xi=\alpha}^{\xi=0} \notag \\
  & = \frac{1}{h}\big( \log(2h) - \log(2(1-h))\big).\notag
\end{align}
Here, we have used \eqref{eq:app3} and \eqref{eq:app4} for the second
$\approx$. In addition,
\begin{eqnarray*}
\lefteqn{4h A_2 - \frac{1}{1-h}(\log(2(1-h)\alpha)+\gamma)}\quad
\\
& \approx&  4h\int_0^\alpha \frac{\int_0^\xi e^{-2h\eta +
      (2h-1)\eta^2/\alpha}d\eta \int_\xi^\alpha e^{-2h\eta +
      (2h-1)\eta^2/\alpha}d\eta}{(\alpha-\xi) e^{-2h\xi +
      (2h-1)\xi^2/\alpha}}d\xi
\\
&& \qquad {} - \frac{1}{1-h}
  \int_0^{2(1-h)\alpha}\frac{1-e^{-\xi}}{\xi}d\xi
\\
& \approx&
  \frac{1}{1-h} \Big(\int_0^\alpha \frac{1 - e^{-(2h -
      (2h-1)2\xi/\alpha)(\alpha-\xi)}}{\alpha-\xi}\frac{2(1-h)}{2h -
    2(2h-1)\xi/\alpha} d\xi 
\\
&& \qquad\qquad {} - \int_0^\alpha
  \frac{1-e^{-2(1-h)\xi}}{\xi}d\xi\Big)
\\
& \stackrel{\xi\to\alpha-\xi}\approx& \frac{1}{1-h} \Big( \int_0^\alpha
  \frac{1-e^{-2(1-h) \xi}}{\xi} \frac{2(1-h)}{2(1-h) -
    2(2(1-h)-1)\xi/\alpha}d\xi 
\\
&& \qquad\qquad {} -
  \int_0^{\alpha}\frac{1-e^{-2(1-h)\xi}}{\xi}d\xi \Big)
\end{eqnarray*}
The last line is exactly line \eqref{eq:app6} with $h$ replaced by
$1-h$, and hence we obtain
\begin{align*}
  \mathbb E_0^{\alpha, h}[T^\ast] - \frac{\log(2\alpha)}{\alpha
    h(1-h)}& \approx \tfrac{1}{\alpha}\big(\tfrac{1}{h}( \gamma +
  2\log h - \log(1-h))\\ & \qquad \quad  +
  \tfrac{1}{1-h}( \gamma + 2\log(1-h) - \log h)\big) \\ & =
  \frac{1}{\alpha h(1-h)} \Big( \gamma + (2-3h)\log h
  + (3h-1)\log(1-h)\Big).
\end{align*}

The variance is computed by \eqref{eq:green2} and
\begin{align*}
 \mathbb V_0^{\alpha, h}[T^\ast]  &= 8 \int_0^1 \int_0^\xi
  \frac{\Big(\int_\xi^1 e^{-2h\alpha\zeta + (2h-1)\alpha
      \zeta^2}d\zeta\Big)^2 }{\xi(1-\xi)\eta(1-\eta)
    e^{-2h\alpha\xi + (2h-1)\alpha \xi^2}e^{-2h\alpha\eta +
      (2h-1)\alpha\eta^2} }  
\\
&  \qquad \qquad \cdot \frac{\Big(\int_0^\eta e^{-2h\alpha\zeta +
      (2h-1)\alpha\zeta^2}d\zeta\Big)^2}{\Big(\int_0^1 e^{-2h\alpha\zeta +
      (2h-1)\alpha\zeta^2}d\zeta\Big)^2}d\eta d\xi
\\
& = 8
  \int_0^\alpha \int_0^\xi
\frac{\Big(\int_\xi^\alpha e^{-2h\zeta +
      (2h-1) \zeta^2/\alpha} d\zeta\Big)^2}{\xi(\alpha-\xi)\eta(\alpha-\eta)
    e^{-2h\xi + (2h-1) \xi^2/\alpha}e^{-2h\eta +
      (2h-1)\eta^2/\alpha}} 
\\
& \qquad \qquad \cdot \frac{ \Big(\int_0^\eta
    e^{-2h\zeta +
      (2h-1)\zeta^2/\alpha}d\zeta\Big)^2}{ \Big(\int_0^\alpha e^{-2h\zeta +
      (2h-1)\zeta^2/\alpha}d\zeta\Big)^2}d\eta d\xi
\\
& \approx
  \frac{8}{\alpha^2} \Big( \int_0^{c_\alpha} \int_0^\xi
  \frac{e^{-4h\xi} \Big(\int_0^\eta
    e^{-2h\zeta}d\zeta\Big)^2}{\xi \eta e^{-2h\xi}e^{-2h\eta}}d\eta d\xi
\\*
& \qquad \qquad  +
  \int_{\alpha-c_\alpha}^\alpha \int_{\alpha-c_\alpha}^\xi
  \frac{1}{(\alpha-\xi)(\alpha-\eta)
    e^{-2h\xi + (2h-1) \xi^2/\alpha}}
\\
& \qquad \qquad \qquad\qquad   \cdot
    \frac{\Big(\int_\xi^\alpha e^{-2h\zeta + (2h-1)
      \zeta^2/\alpha}d\zeta\Big)^2}{e^{-2h\eta +
      (2h-1)\eta^2/\alpha}}d\eta d\xi\Big)
\\
& \approx
  \frac{8}{\alpha^2}\Big( \int_0^{c_\alpha} \int_0^\xi
  \frac{e^{-2h(\xi-\eta)} \Big(\int_0^\eta
    e^{-2h\zeta}d\zeta\Big)^2}{\xi \eta}d\eta d\xi
\\
& \qquad \qquad  + \int_0^{c_\alpha}
  \int_\xi^{c_\alpha} \frac{\Big(\int_0^\xi
    e^{2(1-h)\zeta}d\zeta\Big)^2}{\xi\eta
    e^{2(1-h)\xi}e^{2(1-h)\eta}}d\eta d\xi\Big)
\\
& \approx
  \frac{2}{\alpha^2}\Big( \frac 1{h^2}\int_0^{\infty} \int_0^\xi
  \frac{e^{-(\xi-\eta)} \Big(\int_0^\eta
    e^{-\zeta}d\zeta\Big)^2}{\xi \eta}d\eta d\xi
\\ & \qquad \qquad + \frac{1}{(1-h)^2} \int_0^{\infty}
  \int_\xi^{\infty} \frac{\Big(\int_0^\xi
    e^{\zeta}d\zeta\Big)^2}{\xi\eta e^{\xi}e^{\eta}}d\eta d\xi\Big)
\\
  & \approx \frac{1}{\alpha^2}\Big( \frac{c'}{h^2} +
  \frac{c''}{(1-h)^2}\Big)
\end{align*}
for some finite $c', c''$.

\noindent 2. Now, we come to the case $h=1$. Here, we use the scale
function (which is up to a factor of $e^{-\alpha}$ the same as in
\eqref{eq:scale} for $h=1$)
$$ S(x) = \int_0^x \exp\Big( - 2\alpha\int_1^\xi(1-\eta)d\eta\Big)d\xi = \int_{1-x}^{1}
e^{\alpha\xi^2}d\xi.$$ Then,
\begin{align*}
  \mathbb E_0^{\alpha, h=1}[T^\ast] & =2\int_0^1 \frac{\int_0^\xi e^{
      \alpha \eta^2}d\eta \int_\xi^1e^{ \alpha
      \eta^2}d\eta}{\xi(1-\xi) e^{ \alpha \xi^2} 
    \int_0^1e^{ \alpha \eta^2} d\eta} d\xi\\ & = \frac{2}{\sqrt\alpha}
  \int_0^{\sqrt\alpha} \frac{\int_0^\xi e^{\eta^2}d\eta
    \int_\xi^{\sqrt\alpha}e^{\eta^2}d\eta}{\xi(1-\xi/\sqrt\alpha)
    e^{\xi^2}  \int_0^{\sqrt\alpha}e^{\eta^2} d\eta} d\xi \\ & =
  \frac{2}{\sqrt\alpha}(A + \Delta_1 + \Delta_2+ \Delta_3 + \Delta_4),
\end{align*}
such that (recall \eqref{eq:app2}), and using $\int_\xi^{\sqrt\alpha}
e^{\eta^2} d\eta = \int_0^{\sqrt\alpha} e^{\eta^2} d\eta - \int_0^{\xi}
e^{\eta^2} d\eta$ and {\sc Mathematica} for the first term
\begin{align*}
  A & = \int_0^\infty \frac{\int_0^\xi e^{\eta^2}d\eta}{\xi
    e^{\xi^2}}d\xi = \frac{\pi^{3/2}}{4},\\
  \Delta_1 & = - \int_{\sqrt\alpha}^\infty \frac{\int_0^\xi
    e^{\eta^2}d\eta}{\xi e^{\xi^2}}d\xi \approx -
  \int_{\sqrt\alpha}^\infty \frac{1}{2\xi^2} d\xi = -
  \frac{1}{2\sqrt\alpha}, \\\Delta_2 & = -\int_0^{\sqrt\alpha}
  \frac{\int_0^{\xi} e^{\eta^2} d\eta  \int_0^{\xi} e^{\eta^2}
    d\eta}{\xi e^{\xi^2} \int_0^{\sqrt\alpha}e^{\eta^2} d\eta}d\xi
  \approx - \frac{2\sqrt{\alpha}}{e^\alpha}
  \int_{c_\alpha}^{\sqrt\alpha}
  \frac{e^{\xi^2}}{4\xi^3}\frac{8\xi^4}{8\sqrt{\alpha}^4} d\xi
  \approx-
  \frac{1}{4\sqrt{\alpha}^3},\\
  \Delta_3 & = \int_0^{\sqrt{\alpha}-c_\alpha}
  \frac{\int_\xi^{\sqrt\alpha}e^{\eta^2}d\eta\int_0^\xi
    e^{\eta^2}d\eta}{(\sqrt\alpha - \xi) e^{\xi^2} 
    \int_0^{\sqrt\alpha}e^{\eta^2} d\eta} d\xi \approx
  \int_{c_\alpha}^{\sqrt \alpha - c_\alpha} \frac{1}{2\xi(\sqrt\alpha
    - \xi) }d\xi \\ & = \frac{1}{2\sqrt\alpha}
  \int_{c_\alpha}^{\sqrt\alpha - c_\alpha} \frac{1}{\xi} + \frac
  1{\sqrt\alpha - \xi} d\xi \approx
  \frac{\log\alpha}{2\sqrt\alpha},\\
  \Delta_4 & = \int_{\sqrt{\alpha}-c_\alpha}^{\sqrt\alpha}
  \frac{\int_\xi^{\sqrt\alpha}e^{\eta^2}d\eta\int_0^\xi
    e^{\eta^2}d\eta}{(\sqrt\alpha - \xi) e^{\xi^2} 
    \int_0^{\sqrt\alpha}e^{\eta^2} d\eta} d\xi \approx
  \frac{1}{e^\alpha} \int_{\sqrt{\alpha}-c_\alpha}^{\sqrt\alpha}
  \frac{\int_\xi^{\sqrt\alpha}
    \frac{2\eta}{2\sqrt\alpha}e^{\eta^2}d\eta}{\sqrt\alpha - \xi}d\xi
  \\ & = \frac{1}{2\sqrt\alpha}
  \int_{\sqrt{\alpha}-c_\alpha}^{\sqrt\alpha} \frac{1 -
    e^{-(\alpha-\xi^2)}}{\sqrt\alpha - \xi}d\xi =
  \frac{1}{2\sqrt\alpha}\int_{0}^{c_\alpha} \frac{1 -
    e^{-\xi(2\sqrt\alpha - \xi)}}{\xi}d\xi \\ & \approx
  \frac{1}{2\sqrt\alpha}\int_0^{c_\alpha\sqrt\alpha} \frac{1 -
    e^{-2\xi}}{\xi} d\xi \approx \frac{\log\alpha}{4\sqrt\alpha}.
\end{align*}
For the variance,
\begin{eqnarray}
  \label{eq:var2}
    \alpha \mathbb V_0^{\alpha, h=1}[T^\ast] %
& = & 8\alpha \int_0^1
    \int_0^\xi \frac{\Big(\int_0^{1-\xi} e^{\alpha
        \zeta^2}d\zeta\Big)^2 \Big(\int_{1-\eta}^1 e^{\alpha
        \zeta^2}d\zeta\Big)^2}{\Big(\int_0^{1} e^{\alpha
        \zeta^2}d\zeta\Big)^2  \xi(1-\xi)\eta(1-\eta)
      e^{\alpha(1-\xi)^2}e^{\alpha(1-\eta)^2}} d\eta d\xi
\nonumber\\
& = & 8
    \alpha \int_0^{\sqrt\alpha} \int_\xi^{\sqrt\alpha}
    \frac{\Big(\int_0^{\xi} e^{\zeta^2}d\zeta\Big)^2
      \Big(\int_{\eta}^{\sqrt\alpha} e^{
        \zeta^2}d\zeta\Big)^2}{\Big(\int_0^{\sqrt\alpha} e^{
        \zeta^2}d\zeta\Big)^2 
      \xi(\sqrt\alpha-\xi)\eta(\sqrt\alpha-\eta) e^{\xi^2}e^{\eta^2}}
    d\eta d\xi %
\nonumber\\
& \approx& 8 \int_0^{\infty} \int_0^\eta
    \frac{\Big(\int_0^{\xi} e^{\zeta^2}d\zeta\Big)^2}{\xi\eta
      e^{\xi^2}e^{\eta^2}} d\xi d\eta < \infty.
\end{eqnarray}

\bigskip\noindent%
\textbf{Proof of Theorem \ref{T2h}}\\
The proof uses an approach similar to the proof of Proposition 3.4 in
\cite{EtheridgePfaffelhuberWakolbinger2006}. Recall the structured
coalescent from Definition \ref{def:K}. We need to show the following
assertions in the limit of large $\alpha$:
\begin{enumerate}[(i)]
\item[(i)] any line never undergoes a transition 4. (back
  recombination to the beneficial background).
\item[(ii)] any pair of lines never undergoes a transition
  2. (coalescence in the wild-type background).
\item[(iii)] any pair of lines never makes a transition 1.\
  (coalescence in the beneficial background) and the coalesced line
  then makes a transition 3. (recombination to the wild-type).
\item[(iv)] the probability for each line to stay in the beneficial
  background (i.e. no transition 3., recombination to the wildtype
  background) until the origin of the beneficial allele is
  $e^{-\lambda/h}$.
\item[(v)] if two lines do not coalesce (no transition 1.),
  recombination to the wildtype background (transition 3.) occurs for
  both lines independently.
\end{enumerate}
Let us explain how these five assertions imply Theorem \ref{T2h}.
Consider a single line, i.e.\ $\mathcal K_0=(1,0)$. When (i),...,(v)
are shown, (i) and (iv) immediately imply \eqref{T21h}. Next, consider
$\mathcal K_0=(2,0)$. For \eqref{T22h}, the event $K_{T^\ast}^b +
K_{T^\ast}^w=1$ means that an event 1.\ has occurred (since 2.\ does
not occur by (ii)). In addition, the lines coalesce at $\beta=T^\ast$,
the beginning of the sweep. Otherwise there would be the chance that
the coalesced line recombines through a transition 3., which is not
possible by (iii). In particular, the last two arguments show that
$\mathcal K_{T^\ast}=(1,0)$. Hence, $\mathcal K_{T^\ast}=(1,0)$ iff
both lines do not make a transition 3.\ if coalescence cannot
occur. This probability is $e^{-2\lambda/h}$ by (v).

~

The main point in showing the star-like approximation is (iii), as
this assertion exactly implies that the structured coalescent
converges to a star-like tree. So we start with proving this. Since
transitions 1.\ occur at rate $1/X_\beta$ and transitions 3.\ at rate
$\rho(1-X_\beta)$, and using the Green function identity from \citet[equation 2.1.1]{DawsonGorostizaWakolbinger2001} in the third line,
\begin{eqnarray}
  \label{eq:T25}
\lefteqn{\mathbb P [\text{event described in (iii) occurs}]}
\qquad \nonumber\\
& =&  \mathbb
    E^{\alpha, h} \Big[ \int_0^{T^\ast} \int_0^t \rho(1-X_s) \exp\Big(
    - \int_s^{T^\ast} \rho(1+1_{r>t})(1-X_r)dr\Big)
\nonumber \\
&& \qquad\qquad
\cdot \frac{1}{X_t}
    \exp\Big( - \int_t^{T^\ast} \frac{1}{X_r} dr\Big)ds dt\Big]
\nonumber\\
& \leq &
\mathbb E^{\alpha, h} \Big[ \int_0^{T^\ast} \int_0^t
    \rho(1-X_s) \frac{1}{X_t}ds dt\Big]
\nonumber \\
&=& \rho \int_0^1 \int_0^1
    G^\ast (0,\xi) G^\ast (\xi,\eta) (1-\xi) \frac{1}{\eta} d\eta d\xi 
\nonumber\\
& \approx& \frac{4\lambda \alpha}{\log\alpha}(B_1+B_2),
\end{eqnarray}
where $B_1$ and $B_2$ are given below. For $B_1$, we have for some finite $c$, using \eqref{eq:greenX},
\begin{eqnarray}
  \label{eq:T26}
    B_1 & =& \frac 14 \int_0^1 \int_0^\xi G^\ast(0,\xi)G^\ast
    (\xi,\eta) (1-\xi) \frac{1}{\eta} d\eta d\xi
\nonumber\\
& =& \int_0^1
    \int_0^\xi \frac{\Big(\int_\xi^1 e^{-2h\alpha\zeta + (2h-1)\alpha
        \zeta^2}d\zeta\Big)^2 }{\xi\eta^2(1-\eta)
      e^{-2h\alpha\xi + (2h-1)\alpha \xi^2}e^{-2h\alpha\eta +
        (2h-1)\alpha\eta^2} }
\nonumber\\
&& \qquad \qquad \qquad \cdot \frac{\Big(\int_0^\eta e^{-2h\alpha\zeta
        + (2h-1)\alpha\zeta^2}d\zeta\Big)^2}{ \Big(\int_0^1 e^{-2h\alpha\zeta +
        (2h-1)\alpha\zeta^2}d\zeta\Big)^2}d\eta d\xi
\nonumber\\
& =&
    \int_0^{2h\alpha} \int_0^\xi \frac{\Big(\int_\xi^{2h\alpha}
      e^{-\zeta + (2h-1) \zeta^2/(4h^2\alpha)}d\zeta\Big)^2 }{\xi\eta^2(2h\alpha-\eta)
      e^{-\xi + (2h-1) \xi^2/(4h^2\alpha)}e^{-\eta +
        (2h-1)\eta^2/(4h^2\alpha)} }
\nonumber\\
 &&  \qquad \qquad \qquad \cdot
      \frac{ \Big(\int_0^\eta e^{-\zeta +
        (2h-1)\zeta^2/(4h^2\alpha)}d\zeta\Big)^2}{  \Big(\int_0^{2h\alpha}
      e^{-\zeta + (2h-1) \zeta^2/(4h^2\alpha)} d\zeta\Big)^2}d\eta d\xi
\nonumber\\
    & \approx& \frac{1}{2h\alpha}\int_0^{\infty} \int_0^\xi
    \frac{e^{-(\xi-\eta)} \Big(\int_0^\eta
      e^{-\zeta}d\zeta\Big)^2}{\xi \eta^2}d\eta d\xi \approx
    \frac{c}{2h\alpha}.
\end{eqnarray}
For $B_2$, we compute, with some finite $c'$,
\begin{eqnarray}
  \label{eq:T27}
B_2 & =& \frac 14 \int_0^1 \int_\xi^1 G^\ast(0,\xi)G^\ast
    (\xi,\eta) (1-\xi) \frac{1}{\eta} d\eta d\xi
\nonumber\\
& = &\int_0^1
    \int_\xi^1 \frac{\int_0^\xi
      e^{-2h\alpha\zeta+\alpha\zeta^2(2h-1)}d\zeta  \int_\xi^1
      \cdots d\zeta  \int_0^\eta \cdots d\zeta  \int_\eta^1
      \cdots d\zeta}{\xi\eta^2(1-\eta)
      e^{-2h\alpha\xi+\alpha\xi^2(2h-1)}e^{-2h\alpha\eta+\alpha\eta^2(2h-1)}
       \Big(\int_0^1 \cdots d\zeta\Big)^2} d\eta d\xi
\nonumber\\
& =& \int_0^{2h\alpha} \int_\xi^{2h \alpha} \frac{\int_0^\xi e^{-\zeta
        + \zeta^2(2h-1)/(4h^2\alpha)}d\zeta  \int_\xi^{2h\alpha}
      \cdots d\zeta  \int_0^\eta \cdots d\zeta}{\xi\eta^2(2h\alpha-\eta)e^{-\xi +
        \xi^2(2h-1)/(4h^2\alpha)}e^{-\eta + \eta^2(2h-1)/(4h^2\alpha)}} 
\nonumber\\
& & \qquad \qquad \qquad \cdot \frac{
      \int_\eta^{2h\alpha} \cdots
      d\zeta}{       \Big(\int_0^{2h\alpha} \cdots d\zeta\Big)^2}d\eta d\xi
\nonumber\\
    & \approx& \frac{1}{2h\alpha} \int_0^{c_\alpha} \int_\xi^{c_\alpha}
    \frac{\int_0^\xi e^{-\zeta}d\zeta  \int_\xi^{2h\alpha}
      e^{-\zeta} d\zeta  \int_0^\eta e^{-\zeta} d\zeta 
      \int_\eta^{2h\alpha} e^{-\zeta} d\zeta}{\xi\eta^2 e^{-\xi}
      e^{-\eta} \Big( \int_0^\infty e^{-\zeta}d\zeta\Big)^2} d\eta
    d\xi
\nonumber\\
& \approx& \frac{1}{2h\alpha} \int_0^{c_\alpha} \int_0^\eta
    \frac{\int_0^\xi e^{-\zeta}d\zeta  \int_0^\eta e^{-\zeta}
      d\zeta }{\xi\eta^2 }d\xi d\eta = \frac{c'}{2h\alpha}
\end{eqnarray}
where each integrand in $\int \cdots d\zeta$ is the same as the first
integrand in the respective line. Combining \eqref{eq:T25},
\eqref{eq:T26} and \eqref{eq:T27} shows that
$$ \mathbb P[\text{event described in (iii) occurs}] \xrightarrow{\alpha\to\infty}0$$
which proves (iii). Similar calculations show that
\begin{eqnarray*}
  \lefteqn{\mathbb P [\text{event described in (i) occurs}] }\qquad
\\* &\leq& \rho^2
  \int_0^1 \int_0^1 G^\ast(0,\xi) G^\ast(\xi, \eta)\xi (1-\eta)d\eta
  d\xi \xrightarrow{\alpha\to\infty} 0,
\\
  \lefteqn{\mathbb P [\text{event described in (ii) occurs}]} \qquad
\\* &\leq& \rho
  \int_0^1 \int_0^1 G^\ast(0,\xi) G^\ast(\xi, \eta) \frac{1}{1-\xi}
  (1-\eta)d\eta d\xi \xrightarrow{\alpha\to\infty} 0,
\end{eqnarray*}
and (i), (ii) follow. For (iv) and (v), note that the probability that
a line does not recombine at all is
\begin{align*}
  \mathbb P& [\text{event described in (iv) occurs}] = \mathbb E\Big[
  \exp\Big( - \rho\int_0^{T^\ast} (1-X_s) ds\Big)\Big]
\end{align*}
and (v) is equivalent to
\begin{eqnarray*}
\lefteqn{\mathbb P [\text{neither of two lines recombines}] }\qquad
\\ & =& \mathbb E\Big[
  \exp\Big( - 2\rho\int_0^{T^\ast} (1-X_s) ds\Big)\Big] \approx
  e^{-2\lambda/h}.
\end{eqnarray*}
Hence, (iv) and (v) are proved once we show that
\begin{align}\label{eq:tosee}
  \mathbb E\Big[ \exp\Big( - \rho\int_0^{T^\ast} (1-X_s) ds\Big)\Big]
  \approx e^{-\lambda/h}.
\end{align}
In order to see this, we show that
$$
\rho\int_0^{T^\ast} (1-X_s) ds \xrightarrow{\alpha\to\infty} \frac{\lambda}{h}
$$
in $L^2$. Since $L^2$-convergence implies convergence in distribution, and $x\mapsto e^{-x}$ is bounded on $\mathbb R_+$ and continuous, \eqref{eq:tosee} then follows. We compute
\begin{eqnarray}
  \label{eq:tosee2}
\lefteqn{\mathbb E\Big[\rho  \int_0^{T^\ast} (1-X_s) ds\Big] =  \rho
     \int_0^1 G^\ast(0,\xi) (1-\xi) d\xi }\quad
\nonumber\\
& =&
    \frac{2\rho}{\alpha}\int_0^\alpha \frac{\displaystyle\int_0^\xi
      e^{-2h \eta-(1-2h)\eta^2/\alpha}d\eta\int_\xi^\alpha
      e^{-2h\eta-(1-2h)\eta^2/\alpha}d\eta}{\displaystyle \xi
      e^{-2h\xi-(1-2h)\xi^2/\alpha} \int_0^\alpha
      e^{-2h\eta-(1-2h)\eta^2/\alpha}d\eta}d\xi
\nonumber\\
& \approx &
    \frac{4h\rho}{\alpha} A_1 \approx \frac{\rho\log\alpha}{h\alpha}
    \approx \frac{\lambda}{h}
\end{eqnarray}
with $A_1$ from \eqref{eq:T3p1}, where we have used
\eqref{eq:app6}. Next,
\begin{eqnarray}
  \label{eq:tosee3}
    \mathbb V\Big[\rho^2  \int_0^{T^\ast} (1-X_s) ds\Big] & = & 8
    \rho^2 \int_0^1 \int_0^\xi G^\ast(0,\xi) G^\ast(\xi,\eta)
    (1-\xi)(1-\eta) d\eta d\xi
\nonumber\\ &\leq& 8\rho^2 \mathbb V[T^\ast]
    \xrightarrow{\alpha\to\infty} 0
  \end{eqnarray}
by Theorem \ref{T3}. Combining \eqref{eq:tosee2} and
\eqref{eq:tosee3}, $L^2$ convergence follows and we have proved
\eqref{eq:tosee} as well as Theorem \ref{T2h}.

\begin{acknowledgements}
We thank Cornelia Borck for fruitful discussion, two anonymous
referees for various clarifying comments and Reinhard B\"urger for
bringing the contents of Remark \ref{rem:conv1} to our knowledge.

The work was made possible with financial support by the Vienna
Science and Technology Fund (WWTF) to JH and GE and by the Federal
Ministry of Education and Research, Germany (BMBF) to PP (Kennzeichen
0313921). JH and PP acknowledge support from the DFG Forschergruppe
1078 \emph{Natural selection in structured populations}.
\end{acknowledgements}


\end{document}